\documentclass{LMCS}

\def\dOi{10(3:2)2014}
\lmcsheading%
{\dOi}
{1--43}
{}
{}
{Mar.~17, 2011}
{Aug.~15, 2014}
{}

\ACMCCS{[{\bf Theory of computation}]: Logic; Theory and algorithms
  for application domains---Database theory---Logic and databases}


\usepackage{amsmath}
\usepackage{amssymb}
\usepackage{graphicx,hyperref}

\newtheorem{observation}[thm]{Observation}
\newtheorem*{thm*}{Theorem}

\newcommand{\Cn}[1]{\expandafter\newcommand\csname#1\endcsname}

\Cn{cl}[2]{\lVert #1\rVert_{#2}}

\Cn{makemathname}[1]{\mathrm{#1}}
\Cn{mathname}[2]{\Cn{#1}{\makemathname{#2}}}
\Cn{mathownname}[1]{\mathname{#1}{#1}}

\Cn{calletter}[1]{\Cn{c#1}{{\mathcal #1}}}
\calletter{A}
\calletter{C}
\calletter{F}
\calletter{I}
\calletter{P}
\calletter{T}
\calletter{W}

\Cn{frakletter}[1]{\Cn{f#1}{{\mathfrak #1}}}
\frakletter{A}
\frakletter{B}
\frakletter{C}
\frakletter{L}
\frakletter{S}
\frakletter{T}
\frakletter{U}

\Cn{AB}{\rand{Achim}}
\Cn{MO}{\rand{Martin}}
\Cn{mw}{\rand{Mark}}
\Cn{add}{\leftthreetimes}
\Cn{app}[2]{{#1}\parlr{#2}}
\Cn{arglr}[1]{\lr{[}{#1}{]}}
\Cn{bigset}[2]{\bigl\{\,#1\bigm|#2\,\bigr\}}
\Cn{bq}[1]{\lq{}#1\rq{}}
\Cn{card}{\midlr}
\Cn{code}[1]{\arglr{#1}}
\Cn{defn}{\emph}
\mathownname{DTIME}
\mathownname{NTIME}
\mathownname{DSPACE}
\Cn{eps}{\varepsilon}
\mathownname{FO}
\mathownname{MSO}
\mathownname{GSO}
\mathownname{GF}
\mathownname{ML}
\mathownname{SAT}
\mathownname{BDD}
\mathownname{fin}

\Cn{BDDm}{\BDD^1}
\Cn{Lmu}{\mathrm L_\mu}
\Cn{AFO}{\FO_-(\forall^\ast)}
\Cn{EFO}{\FO_+(\exists^\ast)}
\Cn{muGF}{\mu\GF}
\Cn{GFs}{\GF^\ast}
\Cn{muGFs}{\muGF^\ast}
\Cn{GSOg}{\mathrm{GGSO}}
\Cn{GSOgs}{\GSOg^\ast}
\Cn{GSOs}{\GSO^\ast}
\Cn{gdd}{\mathrm{gdd}}
\Cn{len}{\midlr}
\Cn{lp}{\triangleright}
\Cn{Lpos}[1]{\MSO^{#1}_X[\tau]}
\Cn{lr}[3]{\left#1#2\right#3}
\Cn{midlr}[1]{\lr{\lvert}{#1}{\rvert}}
\Cn{mtype}[2]{\tp^{#1}_X(#2)}
\Cn{mType}[1]{\Tp^{#1}_X[\tau]}
\Cn{parlr}[1]{{\lr{(}{#1}{)}}}
\renewcommand{\phi}{\varphi}
\mathownname{poly}
\mathownname{qa}
\mathownname{qr}
\Cn{rand}[2]{$\bullet$\marginpar{{\bf #1:} #2}}
\Cn{rest}[2]{{#1\upharpoonright#2}}
\Cn{seq}{{\,=\,}}
\Cn{set}[2]{\{\,#1\mathrel{|}#2\,\}}
\Cn{setlr}[1]{{\lr{\{}{#1}{\}}}}
\Cn{size}[1]{\lr{\|}{#1}{\|}}
\Cn{subst}[3]{#1[#2/#3]}
\Cn{taux}{\tau\cup\setlr x}
\Cn{tauX}{\tau\cup\setlr X}
\Cn{tauXx}{\tau\cup\setlr{X,x}}
\mathownname{tp}
\mathownname{Tp}

\Cn?{\kern0.08em}

\Cn{mi}[1]{{\triangleleft}#1}
\Cn{mo}[1]{{\triangleright}#1}

\begin{document}

\title[Decidability Results for the Boundedness Problem]{Decidability Results for the Boundedness Problem}

\author[A.~Blumensath]{Achim Blumensath\rsuper a}
\address{{\lsuper{a,b}}Fachbereich Mathematik, Technische Universit\"at Darmstadt}
\email{\{blumensath,otto\}@mathematik.tu-darmstadt.de}

\author[M.~Otto]{Martin Otto\rsuper b}
\address{\vspace{-18 pt}}

\author[M.Weier]{Mark Weyer\rsuper c}
\address{{\lsuper c}\vspace{-12 pt}}
\email{mark@weyer-zuhause.de}

\begin{abstract}
We prove decidability of the boundedness problem for
monadic least fixed-point recursion based on
positive monadic second-order ($\MSO$) formulae over trees.
Given an $\MSO$-formula $\varphi(X,x)$ that is positive in~$X$, it is decidable
whether the fixed-point recursion based on~$\varphi$ is spurious
over the class of all trees in the sense that there is some uniform
finite bound for the number of iterations $\varphi$~takes
to reach its least fixed point, uniformly across all trees.
We also identify the exact complexity of this problem.
The proof uses automata-theoretic techniques.

This key result extends, by means of model-theoretic interpretations,
to show decidability of the boundedness problem for $\MSO$
and guarded second-order logic ($\GSO$) over the classes of structures
of fixed finite tree-width. Further model-theoretic transfer arguments
allow us to derive major known decidability results for boundedness
for fragments of first-order logic as well as new ones.
\end{abstract}

\maketitle

\section{Introduction}
\label{sect:intro}

In applications one frequently employs tailor-made logics
to achieve a balance between expressive power and algorithmic manageability.
Adding fixed-point operators to weak logics turned out to be a good way
to achieve such a balance.
Think, for example of the addition of transitive closure operators
or more general fixed-point constructs to database
query languages, or of various fixed-point defined reachability
or recurrence assertions to logics used in verification, like 
linear or branching time temporal logics or the modal $\mu$-calculus.
Fixed-point operators introduce a measure of relational recursion and
typically boost expressiveness in the direction of more dynamic and
less local properties. They offer relational recursion based on the 
iteration of relation updates that are definable in the underlying logic.
We here primarily consider monadic least fixed points, based on formulae
$\varphi(X,x)$ that are monotone (positive) in the
monadic recursion variable~$X$.
On a fixed structure~$\fA$,
any such~$\varphi$ induces a monotone operation
$F_\varphi : P \mapsto \set{ a \in \fA }{ \fA \models \varphi(P,a)}$
on monadic relations $P \subseteq A$.
The least fixed point of this operation over~$\fA$,
denoted as $\varphi^\infty(\fA)$, is also the first stationary point
of the monotone, ordinal-indexed iteration sequence of stages $\varphi^\alpha(\fA)$ starting from
$\varphi^0(\fA) := \emptyset$, with updates
$\varphi^{\alpha+1}(\fA) := F_\varphi( \varphi^{\alpha}(\fA))$
and unions in limits. 
The least~$\alpha$ for which
$\varphi^{\alpha+1}(\fA) = \varphi^\alpha(\fA)$ is called the closure
ordinal for this fixed-point iteration on~$\fA$.

For a concrete fixed-point process it may be hard to tell
whether the recursion employed is crucial or whether it is
spurious and can be eliminated.
Indeed this question comes in two versions\?:
(a) one can ask whether a resulting fixed point is also
uniformly definable in the base logic without fixed-point recursion
(purely an expressiveness issue)\?;
(b) one may also be interested to know whether the given fixed-point
iteration terminates within a uniformly bounded finite number
of steps (an algorithmic issue, concerning the dynamics of the
fixed-point recursion rather than its result).

The boundedness problem $\BDD(L,\cC)$
for a class of formulae~$L$
and a class of structures~$\cC$
concerns question~(b)\?: to decide, for a given formula $\varphi \in L$,
whether there is a finite upper bound on its closure
ordinal, uniformly across all structures $\fA \in \cC$.
We call such fixed-point iterations, or~$\varphi$ itself,
\emph{bounded over~$\cC$.}

Interestingly, for first-order logic, as well as for many natural
fragments, the two questions concerning eliminability of least fixed
points coincide -- at least over the class of all structures.
By a classical theorem of
Barwise and Moschovakis~\cite{BarwiseMoschovakis78},
the only way that the fixed point $\varphi^\infty(\fA)$ can be
first-order definable for every~$\fA$, is that there is some finite~$\alpha$
for which $\varphi^\infty(\fA) = \varphi^\alpha(\fA)$
for all~$\fA$. The converse is clear from the fact that the unfolding
of the iteration to any fixed finite depth~$\alpha$ is easily mimicked
in $\FO$.

In other cases -- and even for $\FO$ over other, restricted classes of
structures, e.g., in finite model theory -- the two problems can
indeed be distinct, and of quite independent interest.

We here deal with the boundedness issue.
Boundedness (even classically, over the class of all structures, and
for just monadic fixed points as considered above) is
undecidable for most first-order fragments of interest (see, e.g.,~\cite{HillebrandEtAl95}).
Notable exceptions are monadic boundedness for positive existential
formulae (\textsc{Datalog})~\cite{CosmadakisGaKaVa88},
for modal formulae~\cite{Otto99}, and for
(a restricted class of) universal formulae without equality~\cite{Otto06}.

One common feature of these decidable cases of the boundedness problem
is that the fragments concerned have a kind of tree-model property
(not just for satisfiability in the fragment itself, but also for the
fixed points and for boundedness).
This is obvious for the modal fragment~\cite{Otto99}, but clearly
also true for positive existential $\FO$ (derivation trees for monadic
\textsc{Datalog} programs can be turned into models of bounded tree-width),
and similarly also for the restricted universal fragment in~\cite{Otto06}.

Motivated by this observation,
\cite{KOS}~has made a first significant step
in an attempt to analyse the boundedness problem from the opposite perspective,
varying the class of structures rather than the class of formulae. The hope
is that this approach could go beyond an ad-hoc exposition of the
decidability of the boundedness problem for individual syntactic fragments,
and offer a unified model-theoretic explanation instead.
\cite{KOS}~shows that boundedness is decidable for
\emph{all} monadic fixed points in $\FO$ over the class of all acyclic
relational structures.
Technically \cite{KOS}~expands on modal
and locality-based proof ideas and reductions to
the monadic second-order theory of trees from~\cite{Otto99,Otto06}
that also rest on the availability of a Barwise--Moschovakis
equivalence. These techniques do not seem to extend to either
the class of all trees (where Barwise--Moschovakis fails) or to
bounded tree-width (where certain simple locality criteria fail).

The present investigation offers another step forward in the alternative
approach to the boundedness problem, on a methodologically very
different note. Its most important novel feature may be that it deals 
with a setting where neither locality nor
Barwise--Moschovakis are available. On the one hand, the class of formulae
considered is extended from first-order logic $\FO$
to full monadic second-order logic $\MSO$ -- a leap which greatly increases
the robustness of the results w.r.t.\ interpretations, and hence
their model-theoretic impact. On the other hand, automata are crucially used
and the underlying structures are restricted to trees.
Using $\MSO$-interpretations it follows that
the boundedness problem for $\MSO$ is decidable
over any $\MSO$-definable class of bounded tree-width, and similarly
even for guarded second-order logic $\GSO$ instead of $\MSO$.

These ramifications demonstrate the strength and unifying explanatory
power of our main decidability result in the wider context of the
boundedness issue.
One of our strongest concrete decidability results
concerns the boundedness problem for $\GSO$ over
$\GSO$-definable classes of bounded tree-width, cf.~Corollary~\ref{cor:decidability of BDD(GF), etc}.
This, in its turn,
encompasses all the major, previously known decidability results
for natural fragments of $\FO$ and, furthermore, settles decidability
of boundedness for the guarded fragment $\GF$.
Equally importantly it goes a long way
to explain the perceived dichotomy between the many undecidability
results, which may typically be understood in terms of reductions from
the tiling problem over suitably grid-like structures, and the
comparatively rare cases of decidability, which can now be systematically
linked to some generalised tree-model property.

Among the classical and previously known decidability results
for the boundedness of (systems of) monadic least fixed points,
which can be integrated into this new picture, are those for
\begin{itemize}
\item[--] monadic \textsc{Datalog,} or systems of monadic least fixed points
  for the purely existential--positive fragment of first-order logic,
  \cite{CosmadakisGaKaVa88}\?;
\item[--] dually, (systems of) monadic least fixed points
  in the purely universal-negative fragment of first-order logic
  (which may equivalently be phrased in terms of the
  boundedness for greatest fixed points for \textsc{Datalog} or for
  existential--positive first-order), \cite{Otto06}\?;
\item[--] modal logic, \cite{Otto99}\?;
\item[--] monadic least fixed points for unconstrained $\FO$
  in restriction to the class of all acyclic relational structures, \cite{KOS}.
\end{itemize}

\noindent Our decidability results are based on a reduction of the monadic
boundedness problem to the \emph{limitedness problem} for
\emph{weighted parity automata,} whose decidability is due to
Colcombet and L\"oding \cite{ColcombetLoeding08}
(cf.~Theorems
\ref{thm:finite limitedness decidable}~and~\ref{thm:limitedness decidable}
below).
This reduction introduces a rather sophisticated annotation
(of ternary tree structures) that records dependencies between the stages
of a fixed-point iteration over these tree structures\?; it is
established that, subject to a limitedness condition on a related
cost function, these annotations can serve as certificates for boundedness.

The overall structure of the paper is as follows.
We divide the material into two major parts\?:
the first part, comprising Sections~\ref{sect:start I}--\ref{sect:end I},
is devoted to the development of the
new techniques and leads up to the core technical result\?:
the decidability of the boundedness problem for $\MSO$
on the class of all ternary trees
through reduction to the limitedness problem for a certain class of automata.
The ramifications of this result are investigated
in the second half of the paper. 
Sections~\ref{sect:start II}--\ref{sect:end II} develop
transfer and reduction arguments that allow us to make links with previously
known decidability results and to derive several new concrete
decidability results. 
Section~\ref{sect:complexity}, finally, discusses complexity issues.

\section{Preliminaries}
\label{sect:prelims}

We assume some familiarity with basic concepts of logic
as can be found, e.g., in \cite{EbbinghausFlum95}.
Throughout the paper we assume that all vocabularies are finite
and that they contain only relation symbols and constant symbols,
but no function symbols.

Consider a second-order formula $\varphi(X,\bar x)$
with free variables as indicated in an underlying vocabulary~$\tau$.
Suppose that $\varphi(X,\bar x)$ is positive in the $r$-ary
second-order variable~$X$ and $\bar x = (x_1,\dots,x_r)$ is a matching
tuple of free first-order variables.
Any $X$-positive formula of this format induces, over every
$\tau$-structure~$\fA$, an operation on the power set of $A^r$\?:
\begin{align*}
  P \mapsto \varphi(\fA,P) := \set{ \bar{a} \in A^r }{ (\fA,P,\bar{a}) \models \varphi }\,.
\end{align*}

As $\varphi$~is $X$-positive, this operation is monotone
($P \subseteq P'$ implies $\varphi(\fA,P) \subseteq \varphi(\fA,P')$)
and hence possesses a unique least fixed point,
which we denote as $\varphi^\infty(\fA)$. This least fixed point
is obtained as the limit of the monotone sequence of inductive stages
$\varphi^\alpha(\fA)$
induced by~$\varphi$ on~$\fA$. These stages are defined by transfinite
induction, for all ordinals $\alpha$, according to\?:
\begin{align*}
  \varphi^0(\fA)          &:= \emptyset\,,  \\
  \varphi^{\alpha+1}(\fA) &:= \varphi(\fA,\varphi^\alpha(\fA))\,, \\
  \varphi^\delta(\fA)     &:= \bigcup_{\alpha < \delta} \varphi^\alpha(\fA)
                              \quad\text{for limits } \delta\,.
\end{align*}

The \emph{finite stages} $\varphi^\alpha(\fA)$, for $\alpha < \omega$,
are uniformly definable by formulae, which we also denote by~$\varphi^\alpha$,
obtained from $\varphi(X,\bar x)$ by iterated substitution
of~$\varphi$ for~$X$ in~$\varphi$.
Letting $\varphi[\psi(\bar x)/X]$ stand for the result of
replacing all free occurrences of~$X$ in atoms~$X\bar y$
in~$\varphi$ by~$\psi(\bar y)$, we obtain
formulae~$\varphi^\alpha$ for $\alpha < \omega$, by
\begin{align*}
  \varphi^0 := \bot
  \quad\text{and}\quad
  \varphi^{\alpha+1} := \varphi[\varphi^\alpha(\bar x)/X]\,.
\end{align*}

Clearly, for finite $\alpha$, $\varphi^\alpha \in \MSO$ for $\varphi \in \MSO$,
and similarly for all natural fragments of first- and second-order logic that
are closed under this substitution operation. It is easy to see that
$\varphi^\alpha$~defines the stage $\varphi^\alpha(\fA)$
for finite $\alpha$, uniformly across all~$\fA$. We therefore need not
distinguish between the two readings of $\varphi^\alpha(\fA)$ for
\emph{finite}~$\alpha$. For infinite~$\alpha$, on the other hand,
we do not regard~$\varphi^\alpha$ as a formula (it would in general have to be
a formula in some infinitary extension of the base logic), but only
allow $\varphi^\alpha(\fA)$ as shorthand notation for the corresponding stage
of~$\varphi$ over~$\fA$.

Because of monotonicity,
$\varphi^\infty(\fA) = \bigcup_\alpha \varphi^\alpha(\fA) = \varphi^\gamma(\fA)$
for the least ordinal~$\gamma$ for which
$\varphi^{\gamma+1}(\fA) = \varphi^\gamma(\fA)$.
This ordinal~$\gamma$ is called the \emph{closure ordinal} for~$\varphi$
on~$\fA$, denoted $\cl{\varphi}{\fA}$.
The \emph{stage} of an individual $\bar a \in \varphi^\infty(\fA)$
is the least ordinal~$\alpha$ such that $\bar a \in \varphi^{\alpha+1}(\fA)$\?;
therefore, the closure ordinal could also be described as the least
ordinal greater than the stages of all members of the fixed point $\varphi^\infty(\fA)$.

The closure ordinal can in general only be bounded, for simple
cardinality reasons, by the successor cardinal of the cardinality
of~$\fA$, or by $\card{A}^r$ for finite~$\fA$.

For instance, the fixed-point induction based on $\varphi(X,x) = \forall y (Ryx \to Xy)$
yields as its fixed point over $\fA = (A,R)$ the set of elements
$a \in A$ that are well-founded w.r.t.\ $R$\?;
over the well-ordering $\fA = (\alpha,{<})$, the closure ordinal is~$\alpha$.
In fact, $\varphi^\infty(\fA) = \alpha$\?;
the stage of $\beta \in \alpha$ is~$\beta$.

The fixed-point induction based on~$\varphi$, or for simplicity\?: $\varphi$~itself,
is said to be bounded if, for some finite $\alpha < \omega$,
$\cl{\varphi}{\fA} \leq \alpha$ for all~$\fA$.
Similarly, $\varphi$~is bounded
over the class~$\cC$ if, for some $\alpha < \omega$,
$\cl{\varphi}{\fA} \leq \alpha$ for all $\fA \in \cC$.

\begin{defi}
{\normalfont (a)} Let $\varphi$~be a formula over~$\tau$, positive in~$X$,
and let $\alpha < \omega$.
We say that $\varphi$~is \emph{bounded by~$\alpha$} over a class~$\cC$
if $\varphi^{\alpha}(\fA) = \varphi^{\alpha+1}(\fA)$, for all $\fA \in \cC$.
We call~$\varphi$ \emph{bounded} over~$\cC$
if it is bounded by some $\alpha < \omega$.

{\normalfont (b)} The \emph{boundedness problem} for a logic~$L$
over a class~$\cC$ is the problem to decide,
given a formula $\varphi \in L$, whether $\varphi$~is bounded over~$\cC$.
We denote this decision problem as $\BDD(L,\cC)$.

The \emph{monadic boundedness problem} is
the corresponding problem where we only consider formulae~$\varphi$
with \emph{monadic} variables~$X$.
We denote it as $\BDDm(L,\cC)$.

If $\cC$~is the class of all structures,
we just write $\BDD(L)$ or $\BDDm(L)$.
\end{defi}

A vocabulary~$\tau$ is called a \emph{tree vocabulary,}
if $\tau$~consists of one binary relation symbol~$E$
and, otherwise, only of constant symbols and unary relation symbols.
A $\tau$-structure~$\fT$ is called a \emph{tree structure,}
or \emph{tree} for short,
if $E^\fT$ is a symmetric, acyclic, and connected relation on~$T$.
In particular, tree structures are undirected.

In Part~I of the paper,
we shall exclusively look at $\BDDm(\MSO,\cT)$
for the class of $\MSO$-formulae $\varphi(X,x)$ suitable for
monadic fixed points (positive in the monadic variable~$X$)
over the class~$\cT$ of all tree structures
and some of its subclasses. We refer to this core problem as
the boundedness problem for $\MSO$ over trees for short.

\begin{thm}[Main theorem]
$\BDDm(\MSO,\cT)$, the monadic boundedness problem for $\MSO$
over the class of all tree structures, is decidable.
\end{thm}

In Part~II we employ model-theoretic interpretations and similar transfer arguments
to deduce from this result the decidability of many other boundedness problems.
In particular, we obtain new proofs of many previous decidability results
for boundedness,
as well as some new results like the decidability for the guarded fragment
of first-order logic and for full guarded second-order logic over structures
of bounded tree-width.

\bigskip
\section*{Part I. The main result}

In this first part we prove the main technical result, the decidability
of the monadic boundedness problem for $\MSO$ on the class of all ternary trees.
The ramifications of this result will then be investigated
in the second half of the paper.

\medskip
To help the reader through the later technicalities,
we start with a simplified outline of the proof idea
towards the main theorem.
The key idea is to derive, for every formula~$\varphi$, a
bound $N = N(\varphi)$ that provides a uniform strict upper bound on the
closure ordinals $\cl{\varphi}{\fT}$ over any tree structure~$\fT$ in case
$\varphi$ \emph{is} bounded.
Then boundedness of $\varphi$ is equivalent to the unsatisfiability
of $\varphi^N \wedge \neg \varphi^{N-1}$
(over the class of all tree structures~$\fT$).
In other words, a formula which (on the class of all trees)
is not bounded by this number~$N$ is not bounded at all.
To reason towards such a uniform bound~$N$,
assume that for some tree~$\fT$,
some node~$v$ enters the fixed point in stage~$N$.
Then $(\fT,\varphi^N(\fT),v) \models \varphi$
but $(\fT,\varphi^{N-1}(\fT),v) \not\models \varphi$.
Using a Feferman--Vaught style lemma (cf.~Proposition~\ref{prop:Feferman-Vaught}),
this change in the status of~$\varphi$ can be traced back to some other node~$w$
such that $(\fT,\varphi^N(\fT),w) \models Xx$
but $(\fT,\varphi^{N-1}(\fT),w) \not\models Xx$,
which means that $w$~entered the fixed point in stage $N-1$.
In this way we obtain a path of dependencies
which travels through the tree
and at places decreases the stage by~$1$.
In a chain of $N$~such jumps,
we conclude that, if $N$~is large
in comparison to the number of types used in the Feferman--Vaught style lemma,
then the path has repetitions and we can use a pumping argument
to produce trees where some node enters the fixed point at arbitrarily large stages.
Consequently, $\varphi$~is unbounded.

The actual proof has to deal with further difficulties,
so it does not exactly follow this outline.
One difficulty is that a pumping lemma essentially requires that
(in some very loose sense) we only use regular properties.
In particular, we have to weaken the counting of stages
and, consequently,
we will slightly relax the concept of a dependency.
Also, it is not sufficient to consider a single dependency path\?:
we have to do the pumping such that it works for all paths simultaneously.
Fortunately, there is already a suitable pumping theorem for a certain kind
of weighted automaton that we can reduce our problem to.
The main part of this paper describes this highly non-trivial reduction.

\paragraph*{\itshape Convention.}
For technical reasons we choose in the following not to distinguish
formally between (assignments to) free first and second-order variables
(and interpretations of) constant or relation symbols. For instance,
we shall often regard $x$~and~$X$, which in usual parlance occur
free in $\varphi(X,x)$, as part of the vocabulary, and think of
assignments $a \in A$ and $P \subseteq A$ over some~$\fA$
in terms of the expansion $(\fA,P,a)$ of~$\fA$.

\section{A Feferman--Vaught theorem for positive types}
\label{sect:Feferman-Vaught}
\label{sect:start I}

For a vocabulary~$\tau$, we denote by
$\MSO^n[\tau]$ the set of all $\MSO$-formulae
over~$\tau$ with quantifier rank at most~$n$
(we count both first- and second-order quantifiers).
If $X \in \tau$ is a unary predicate we write
$\Lpos n$ for the subset of all formulae
where the predicate~$X$ occurs only positively.
Recall that, for finite vocabularies~$\tau$, $\MSO^n[\tau]$,
and hence also $\Lpos n$,
is finite up to logical equivalence.

\begin{defi}
Let $\tau$~be a vocabulary and $X \in \tau$.
The \emph{$X$-positive $n$-type} of a $\tau$-structure~$\fA$ is the set
\begin{align*}
  \mtype{n}{\fA} := \set{ \varphi\in\Lpos n }{ \fA\models\varphi }\,.
\end{align*}
We write $\mType{n}$
for the set of all $X$-positive $n$-types of $\tau$-structures.
\end{defi}

Let $\fT_1$ and $\fT_2$ be tree structures.
If $T_1$ and $T_2$ are disjoint,
and if furthermore no constant symbol is interpreted in both trees,
then we define a concatenation operation as follows\?:
let $c_1$~and~$c_2$ be constant symbols from the structures
$\fT_1$~and~$\fT_2$, respectively.
Then we denote by $\fT_1 +_{c_1,c_2} \fT_2$ the tree
obtained from the disjoint union of the trees $\fT_1$~and~$\fT_2$
by adding an edge between $c_1^{\fT_1}$ and $c_2^{\fT_2}$.
Note that every finite tree
can be constructed from one-element trees using this operation and reduct operations.

If $\fT$ is a tree and $vw$ an edge of~$\fT$,
then removing~$vw$ from~$\fT$ produces two disjoint trees.
Of these, we denote the one containing the vertex~$v$ by~$\fT_{vw}$.
Note that, if there are constants $c$~and~$d$ for $v$~and~$w$,
then $\fT = \fT_{vw} +_{c,d} \fT_{wv}$.
If $c$~is a constant symbol not interpreted by~$\fT_{vw}$,
then we set $\fT_{vw,c} := \parlr{\fT_{vw},v}$,
where the expansion interprets~$c$ by~$v$.

We will frequently need a derived operation\?:
let $\fT_1$~and~$\fT_2$ be trees
such that $T_1$ and $T_2$ are disjoint, and
suppose that there is exactly one constant symbol~$c$
that is interpreted both in $\fT_1$~and in~$\fT_2$.
Let $d$~be a constant symbol which is interpreted in neither.
Then we denote by $\fT_1 \add_c \fT_2$
the reduct of $\fT_1 +_{c,d} \subst{\fT_2}{d}{c}$
that expels~$d$ from the vocabulary ($\subst{\fT_2}{d}{c}$ denotes
the structure obtained from~$\fT_2$ by renaming the constant symbol~$c$
to~$d$).
Intuitively,
$c$~denotes the root of (directed versions of) the respective trees,
and $\add_c$~appends its second argument as a new subtree
below the root of its first argument.
For a more uniform treatment,
we allow the empty tree~$\triangle$
as a neutral second argument to~$\add_c$, and
we use~$\triangle$ also for its type.

\begin{prop}\label{prop:Feferman-Vaught}
For every $n < \omega$,
there is a binary operation~$\oplus^n_{c_1,c_2}$ on $X$-positive $n$-types
such that,
for all trees $\fT_1,\fT_2$ for which $\fT_1 +_{c_1,c_2} \fT_2$ is defined,
we have
\begin{align*} \mtype n{\fT_1 +_{c_1,c_2} \fT_2}
    = \mtype n{\fT_1} \oplus^n_{c_1,c_2} \mtype n{\fT_2}\,.
\end{align*}
Furthermore, $\oplus^n_{c_1,c_2}$ is monotone\?:
\begin{align*}
  t_1\subseteq t_1' \text{ and } t_2 \subseteq t_2'
  \quad\text{implies}\quad
  t_1 \oplus^n_{c_1,c_2} t_2 \subseteq t_1' \oplus^n_{c_1,c_2} t_2'\,.
\end{align*}
Finally, $t_1\oplus^n_{c_1,c_2}t_2$ is computable from $n$, $t_1$, and $t_2$.
\end{prop}
\begin{proof}
Computability of the operation will be evident,
once we show how to compute with types in an effective way.
For this sake,
note that we can represent an $n$-type by a finite
set of formulae where all maximal boolean combinations
are in disjunctive normal form without repetition of clauses
or of literals in clauses.

We proceed by induction on~$n$.
Assume that we already know how to compute~$\oplus^m_{c_1,c_2}$ for all $m<n$
and all vocabularies.
For convenience, we set
\begin{align*}
  \fT:=\fT_1 +_{c_1,c_2} \fT_2\,,\quad
  t_1:=\mtype n{\fT_1}\,,\quad
  t_2:=\mtype n{\fT_2}\,,
  \quad\text{and}\quad
  t:=\mtype n{\fT}\,.
\end{align*}
We will describe~$t$ solely in terms of $n$,~$t_1$, $t_2$,
and the operations~$\oplus^m_{c_1,c_2}$ with $m < n$.
Each formula in an $X$-positive $n$-type
is a positive boolean combination of atoms, negated atoms,
and formulae of the form $\exists y\varphi$, $\forall y\varphi$,
$\exists Y\varphi$, and $\forall Y\varphi$,
where $y$~is a first-order variable and $Y$~is a set variable.
Whether the full formula belongs to~$t$ is clearly determined
by whether the individual formulae in the positive boolean combination do.
Also, as the boolean combinations are positive,
monotonicity is preserved.
Hence it suffices to consider subformulae of the above form.

In the following we explicitly treat the cases of
atomic and negated atomic formulae and of
$\exists y\varphi$ and $\forall Y\varphi$.
The remaining cases $\exists Y\varphi$ and $\forall y\varphi$
can be handled using combinations of the techniques
used in these cases.

\smallskip
First, we consider atoms and negated atoms.
Each (negated) atom that only uses constants from~$\fT_i$
occurs in~$t$ iff it occurs in~$t_i$.
It remains to consider (negated) atoms
involving constants from both $\fT_1$ and $\fT_2$.
As $E$~is the only relation symbol of arity more than~$1$,
such an atom must be of the form $c\seq d$ or $Ecd$ where,
without loss of generality, $c$~is from the vocabulary of $\fT_1$
and $d$~from the vocabulary of $\fT_2$.
In this case, we always have $c\seq d \notin t$
and, hence, $\neg(c\seq d) \in t$\?;
so
\begin{align*}
  Ecd\in t \quad\text{iff}\quad \neg Ecd \notin t
  \quad\text{iff}\quad
  c\seq c_1\in t_1 \text{ and } d\seq c_2\in t_2\,.
\end{align*}

\smallskip
Next, let us consider a formula of the form $\exists y\varphi$
with $m := \qr(\varphi) < n$.
We make use of~$\oplus^m_{c_1,c_2}$.
Let $t_1'$ and $t_2'$ be the $X$-positive $m$-types of $\fT_1$ and $\fT_2$,
that is, $t_1' = t_1\cap\MSO^m_X$ and $t_2' = t_2\cap\MSO^m_X$.
Further, let $S_1$ be the set of $X$-positive $m$-types
of expansions of $\fT_1$ by some $a\in T_1$ interpreted for~$y$,
and let $S_2$ be the respective set of types of expansions of $\fT_2$.
Clearly, $\exists y\varphi\in t$ iff $\varphi\in\mtype m{\fT,a}$ for some $a\in T$.
For $a\in T_1$ and $t''_1:=\mtype m{\fT_1,a}$,
the inductive hypothesis implies that
\begin{align*}
  \mtype m{\fT,a}
  &= \mtype m{\fT_1+_{c_1,c_2}\fT_2,\ a} \\
  &= \mtype m{\parlr{\fT_1,a}+_{c_1,c_2}\fT_2} \\
  &= \mtype m{\fT_1,a}\oplus^m_{c_1,c_2}\mtype m{\fT_2}
   = t''_1 \oplus^m_{c_1,c_2} t'_2\,.
\end{align*}
Note that $t''_1\in S_1$.
The case where $a\in T_2$ is similar.
It follows that
\begin{align*}
  \exists y\varphi \in t
  \quad\text{iff}\quad
           & \varphi\in t''_1 \oplus^m_{c_1,c_2} t'_2\,, \text{ for some } t''_1\in S_1\,, \\
\text{or } & \varphi\in t'_1 \oplus^m_{c_1,c_2} t''_2\,, \text{ for some } t''_2\in S_2\,.
\end{align*}
As an artifact of positivity in~$X$,
the set~$S_1$ is not determined by~$t_1$.
The point is that, for instance, if
$\exists x ( Xx \wedge \chi(x)) \in t_1$,
then $S_1$~may or may not contain a type~$t'$
such that $\chi\in t'$ but $Xx\not\in t'$,
because we do not know
about the status of $\exists x(\neg Xx \wedge \chi(x))$.

Unlike $S_1$, the following superset of~$S_1$ is determined by~$t_1$\?:
\begin{align*}
  S_1' := \bigset{ t''_1 \in \mType{m} }{ \textstyle\exists y\bigwedge t''_1 \in t_1 } \supseteq S_1\,.
\end{align*}
(Recall that representations of types are finite,
so $\bigwedge t''_1$ is in fact a formula.)
Hence it suffices to show that
\begin{align*}
  \varphi\in t''_1 \oplus^m_{c_1,c_2} t'_2\,, \text{ for some } t''_1\in S_1
  \quad\text{iff}\quad
  \varphi\in t''_1 \oplus^m_{c_1,c_2} t'_2\,, \text{ for some } t''_1\in S_1'\,.
\end{align*}
(The corresponding statement for~$\fT_2$ then follows by symmetry.)

$(\Rightarrow)$ is trivial since $S_1\subseteq S_1'$.
For $(\Leftarrow)$,
assume that $t''_1$~is a type such that
$\exists y\bigwedge t''_1\in t_1$ and $\varphi\in t''_1\oplus^m_{c_1,c_2}t'_2$.
Let $a\in T_1$ be an element with $(\fT_1,a) \models \bigwedge t''_1$,
and set $t'''_1 := \mtype m{\fT_1,a}$.
Clearly, $t''_1 \subseteq t'''_1$.
Hence, monotonicity of~$\oplus^m_{c_1,c_2}$
implies that $\varphi \in t'''_1\oplus^m_{c_1,c_2}t'_2$, as desired.

It remains to show monotonicity of~$\oplus^n_{c_1,c_2}$
(as far as the formula $\exists y\varphi$ is concerned).
We need to establish that, if $\exists y\varphi\in t_1\oplus^n_{c_1,c_2}t_2$,
then this still holds after increasing $t_1$~or~$t_2$.
This follows from the fact
that the sets $S_1'$~and~$S_2'$ (defined analogously to~$S_1'$)
are monotone in $t_1$ and~$t_2$.

\smallskip
Finally, let us consider a formula of the form $\forall Y\varphi$
with $m := \qr(\varphi) < n$.
This time let $S_1$~be the set of $X$-positive $m$-types
of expansions of~$\fT_1$ by some unary predicate $P\subseteq T_1$ interpreted for~$Y$,
and let $S_2$~be the respective set for~$\fT_2$.
Using the equality
\begin{align*}
  (\fT,P) = (\fT_1,P\cap T_1) +_{c_1,c_2} (\fT_2,P\cap T_2)
\end{align*}
we obtain, similarly to the case above, that
\begin{align*}
  \forall Y\varphi\in t
  \quad\text{iff}\quad
  \varphi\in t''_1\oplus^m_{c_1,c_2}t''_2
  \quad\text{for all } t''_1\in S_1 \text{ and } t''_2\in S_2\,.
\end{align*}

Let us call a pair $(S'_1,S'_2)$
\emph{good for $t_1,t_2$,} if the following conditions hold\?:
\begin{itemize}
\item $S_1'$~is a set of $X$-positive $m$-types
  of the vocabulary used for expansions of $\fT_1$ by~$Y$ and
  $S_2'$~is a corresponding set for of $\fT_2$.
\item $\forall Y\bigvee_{s_1\in S'_1}\bigwedge s_1 \in t_1$
  and $\forall Y\bigvee_{s_2\in S'_2}\bigwedge s_2 \in t_2$.
\item For all $s_1\in S'_1$ and $s_2\in S'_2$
  we have $\varphi\in s_1\oplus^m_{c_1,c_2}s_2$.
\end{itemize}
If $\forall Y\varphi\in t$, then $(S_1,S_2)$ is good, whence a good pair exists.
We claim that the converse also holds, i.e., that
the existence of a good pair implies $\forall Y\varphi \in t$.
Thus, we obtain a characterisation of whether $\forall Y\varphi\in t$
solely in terms of $t_1$,~$t_2$, and~$\oplus^m_{c_1,c_2}$.
Furthermore, being good for $t_1,t_2$
is clearly monotone in $t_1$~and~$t_2$.

To prove the claim, suppose that $\parlr{S'_1,S'_2}$ is a good pair
and let $t''_1\in S_1$ and $t''_2\in S_2$ be arbitrary.
We need to show that $\varphi\in t''_1\oplus^m_{c_1,c_2}t''_2$.
Fix a predicate~$P_1$ such that $t''_1=\mtype m{\fT_1,P_1}$.
By the second condition on good pairs, we have
$\parlr{\fT_1,P_1}\models\bigvee_{s_1\in S'_1}\bigwedge s_1$.
Hence, there is some $s_1\in S'_1$ such that $\parlr{\fT_1,P_1}\models\bigwedge s_1$.
This implies that $s_1 \subseteq t''_1$.
Analogously, we obtain some $s_2\in S'_2$ such that $s_2\subseteq t''_2$.
By the third condition on good pairs, we have $\varphi\in s_1\oplus^m_{c_1,c_2}s_2$.
Therefore, monotonicity of $\oplus^m_{c_1,c_2}$ implies that $\varphi\in t''_1\oplus^m_{c_1,c_2}t''_2$.
\end{proof}

The previous proof needed to consider different vocabularies.
From now on, a single vocabulary nearly suffices.
Let $\tau$ be a fixed tree vocabulary without any constant symbols.
Let $X$ be a unary relation symbol and $x$ a constant symbol
such that $x,X \notin \tau$.
We will consider fixed points with respect to $X$~and~$x$.
The fixed points are evaluated in trees of vocabulary~$\tau$.
Stages of the fixed-point induction are evaluated
in trees of vocabulary $\tauX$.
In order to determine whether a single tree node
belongs to some iteration for the fixed point,
we consider trees of vocabulary $\tauXx$.
If $x$ is present in the vocabulary,
its interpretation can be thought of as the root of the tree.

Let $\varphi$ be a $\tauXx$-formula positive in~$X$
and let $n$ be the quantifier rank of~$\varphi$.

\begin{cor}
Let $y \notin \tauXx$ be a new constant symbol.
We define a binary operation~$\add^n$
on $X$-positive $n$-types of $\tauXx$-structures by
\begin{align*}\textstyle
  s \add^n t :=
  \bigl(s \oplus_{x,y}^n \subst{t}{y}{x}\bigr) \cap \MSO^n_X[\tauXx]\,.
\end{align*}
The operation~$\add^n$ is monotone and satisfies
\begin{align*}
  \mtype n{\fS \add_x \fT} = \mtype n{\fS} \add^n \mtype n{\fT}\,,
\end{align*}
for all non-empty tree structures $\fS$~and~$\fT$ of vocabulary~$\tauXx$.
\qed\end{cor}

We extend $\add^n$ by adjoining the $X$-positive $n$-type~$\triangle$ of the empty tree
as a right-neutral element.
This does not hurt monotonicity\?:
without loss of generality, assume that $n\geq 1$.
Then only~$\triangle$ contains $\forall y\bot$
and only this type does not contain $\exists y\top$,
so it is incomparable to any other type.

In the first part, which contains the technical heart of the article,
we will only consider ternary trees, that is,
undirected trees where each node has degree at most~$3$.
We assume that each such tree~$\fT$ is implicitly equipped
with an edge-colouring using $3$~colours $\{1,2,3\}$.
That means that, for every colour~$d$,
each vertex~$v$ of~$\fT$ has at most one neighbour
that is connected to~$v$ via an edge of colour~$d$.
We call this neighbour ``the neighbour of~$v$ in \emph{direction}~$d$''
and we denote it by~$v^d$.
If there is no such neighbour, we set $v^d := \triangle$.

To account for missing neighbours we extend the above
definition of~$\fT_{vw,x}$
by setting $\fT_{v\triangle,x} := (\fT,v)$
and letting $\fT_{\triangle w,x} := \triangle$.
Furthermore, let $\fT_{\{v\}} := (\fT \restriction \{v\},v)$.
With this notation we have
\begin{align*}
  (\fT,v) = \fT_{\{v\}}
  \add_x \fT_{v^1v,x} \add_x \fT_{v^2v,x} \add_x \fT_{v^3v,x}\,,
\end{align*}
where we assume that the operation~$\add_x$ is associative to the left.

We also need a variant of Proposition~\ref{prop:Feferman-Vaught}
that concerns a decomposition into a possibly infinite number of subtrees.
We omit the proof, which is similar to that of Proposition~\ref{prop:Feferman-Vaught}.
\begin{prop}\label{prop:Feferman-Vaught II}
Let $\fT$ be a $\tauXx$-tree and $\parlr{v_1,d_1},\parlr{v_2,d_2},\ldots$
a sequence of pairwise distinct pairs $\parlr{v_i,d_i}$,
such that $v_i\in T$ and $v_i^{d_i}=\triangle$.
Further, let $\fS_1,\fS_2,\ldots$ and $\fS'_1,\fS'_2,\ldots$
be sequences of $\tauXx$-trees
such that $\mtype n{\fS_i}=\mtype n{\fS'_i}$ for all~$i$.
Finally, let $\fU$~be the tree obtained from~$\fT$
by adding~$\fS_i$ as a child of~$v_i$
in direction~$d_i$ for all~$i$,
and define~$\fU'$ analogously using~$\fS'_i$ instead of~$\fS_i$.
Then, $\mtype n{\fU}=\mtype n{\fU'}$.
\qed\end{prop}

\section{Tilings}
\label{sect:tilings}

We are now in a position to provide a second, more precise proof outline.
Given a tree structure~$\fT$ of vocabulary~$\tau$,
we consider the fixed-point induction of~$\varphi$.
For every stage~$\alpha$ and every vertex~$v$ of~$\fT$
we consider the type $\tp^n_X(\fT,\varphi^\alpha(\fT),v)$.
We annotate~$\fT$ with all these types.
At each vertex~$v$ we write down the list of these types for all stages~$\alpha$.
These annotations can be used to determine the fixed-point rank of elements of~$\fT$.
A vertex~$v$ enters the fixed point at stage~$\alpha$ if the $\alpha$-th entry
of the list is the first one containing a type~$t$ with $Xx \in t$.

We can regard the annotation as consisting of several layers, one for each stage of the
induction. At a vertex~$v$ each change between two consecutive layers is caused
by some change at some other vertex in the previous step. In this way we can trace
back changes of the types through the various layers.

In order to determine whether the fixed-point inductions of the formula are bounded,
we construct a \emph{weighted automaton} (see Section~\ref{sect:automata} below)
that recognises (approximations of) such annotations and that computes
(an approximation of) the length of the longest path of changes in the annotation.

Actually, the annotations we use do not consist of single types but
of tuples of such types, called a \emph{tile.}
In this section we consider single layers of such tiles.
In the next section we will then introduce annotations
consisting of several such layers.

\begin{defi}
A \defn{letter} is a one-element $\taux$-tree.
\end{defi}
Observe that, for each letter~$\fL$,
there are exactly two $\tauXx$-expansions of~$\fL$\?:
one where the element belongs to~$X$ and one where it does not.
Let us denote their $X$-positive $n$-types
by $1_\fL$~and~$0_\fL$, respectively.
Note that $0_\fL \subseteq 1_\fL$ and that $Xx\in 1_\fL \smallsetminus 0_\fL$,
for every $\fL$. We omit the index~$\fL$ whenever it is irrelevant.

We can decompose a $\tauX$-tree~$\fT$ into its one-element
substructures~$\fT_{\{v\}}$, i.e., its letters.
Each of these letters~$\fT_{\{v\}}$ can be labelled with
its type and the types of the subtrees $\fT_{v^dv}$.
\begin{center}
\includegraphics{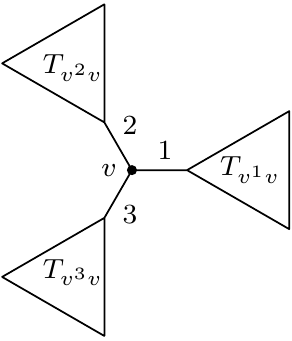}
%
%
%
%
%
%
\end{center}
For convenience, we will not only use the types $t_{\mi0}$ and $t_{\mi d}$
of $\fT_{\{v\}}$ and $\fT_{v^dv}$, $d = 1,2,3$, respectively,
but also the types~$t_{\mo d}$ of $\fT_{vv^d}$, $d = 1,2,3$
and the type~$t_{\mo4}$ of the whole tree $(\fT,v)$.
Our intuition regards the vertex~$v$ as a processing unit
that receives as its inputs the types $t_{\mi 0}$, $t_{\mi 1}$, $t_{\mi 2}$, $t_{\mi 3}$
and produces as output the types $t_{\mo 1}$, $t_{\mo 2}$, $t_{\mo 3}$, $t_{\mo 4}$.
The vertex~$v$ receives from its neighbours~$v^d$, $d = 1,2,3$,
the inputs~$t_{\mi d}$ and it passes back to~$v^d$ the outputs~$t_{\mo d}$.

\begin{defi}
\textup{(a)}
Let $\fL$ be a letter.
An \defn{$\fL$-tile} is an $8$-tuple
\begin{align*}
  \parlr{t_{\mi0},\dots,t_{\mi3},t_{\mo1},\dots,t_{\mo4}}
\end{align*}
of $X$-positive $n$-types over $\tauXx$ where
\begin{itemize}
\item $t_{\mi0}\in\setlr{0_{\fL},1_{\fL}}$,
\item $t_{\mo1} = t_{\mi0} \add^n t_{\mi2} \add^n t_{\mi3}$,
\item $t_{\mo2} = t_{\mi0} \add^n t_{\mi1} \add^n t_{\mi3}$,
\item $t_{\mo3} = t_{\mi0} \add^n t_{\mi1} \add^n t_{\mi2}$, and
\item $t_{\mo4} = t_{\mi0} \add^n t_{\mi1} \add^n t_{\mi2} \add^n t_{\mi3}$.
\end{itemize}
If we do not want to mention the letter,
we refer to an $\fL$-tile simply as a \defn{tile.}
When $\gamma$~is a tile,
we denote its components by $\gamma_{\mi0}$ through $\gamma_{\mo4}$.

\textup{(b)}
Let $\fT$ be a $\tau$-tree.
A \defn{$\fT$-tiling} is a mapping~$c$ that assigns to each
vertex $v \in T$ a $\fT_{\setlr v}$-tile $c(v)$.

\textup{(c)}
Let $\fT$~be a $\tauX$-tree.
The \emph{canonical tiling}~$t_\fT$ of~$\fT$ is the function
assigning to a vertex~$v$ the tile
\begin{alignat*}{-1}
    t_\fT(v)_{\mi 0} &:= \tp_X^n(\fT_{\{v\}})\,, \qquad
  & t_\fT(v)_{\mi d} &:= \tp_X^n(\fT_{v^dv})\,,
  &&\qquad\text{for } 1 \leq d \leq 3\,, \\
    t_\fT(v)_{\mo 4} &:= \tp_X^n(\fT,v)\,, \qquad
  & t_\fT(v)_{\mo d} &:= \tp_X^n(\fT_{vv^d})\,,
  &&\qquad\text{for } 1 \leq d \leq 3\,.
\end{alignat*}
\end{defi}

Intuitively the $\mi d$-component of a tile
contains information \emph{incoming} from direction~$d$,
whereas the $\mo d$-component contains the information passed on in that direction.
Similarly, the $\mo 4$-component contains information passed on to the next stage.
The $\mi 0$-component is special, since it contains local information about the current vertex.

Note that the canonical tiling is indeed a tiling.
\begin{lem}\label{lem:tile is such}
Let $\fT$ be a $\tauX$-tree and $\fT_0$~its $\tau$-reduct.
Then $t_\fT$~is a $\fT_0$-tiling.
\end{lem}
\begin{proof}
Let $v \in T$.
Since $\fT_{\setlr v}$ is an expansion of $\fL := (\fT_0)_{\{v\}}$,
its type $t_\fT(v)_{\mi0}$ must be one of $0_\fL$ and~$1_\fL$.

For the equalities concerning $t_\fT(v)_{\mo d}$ with $1\leq d\leq 3$,
we may by symmetry assume that $d = 3$. Then
\begin{align*}
  t_\fT(v)_{\mo3}
  &= \mtype n{\fT_{vv^3}} \\
  &= \mtype n{
     \fT_{\setlr v}
     \add_x \fT_{v^1v}
     \add_x \fT_{v^2v}} \\
  &= \mtype n{\fT_{\setlr v}}
     \add^n \mtype n{\fT_{v^1v}}
     \add^n \mtype n{\fT_{v^2v}} \\
  &= t_\fT(v)_{\mi0}
     \add^n t_\fT(v)_{\mi1}
     \add^n t_\fT(v)_{\mi2}\,,
\end{align*}
as desired.
The equality for~$\mo4$ is obtained similarly.
\end{proof}

Not every tiling stems from an actual tree.
In the next definition we collect some simple consistency properties
a tiling should satisfy.
Note that these properties can be checked by an automaton.

\begin{defi}
Let $\fT$ be a $\tau$-tree and $v \in T$ a vertex.

\textup{(a)}
The \defn{orientation of~$\fT$ towards $v$}
is the mapping $o_v: T \to \setlr{1,\ldots,4}$
such that $\app{o_v}v=4$ and, for vertices $w\in T\smallsetminus\setlr v$,
we define $1\leq \app{o_v}w\leq 3$ such that the neighbour $w^{o_v(w)}$ is closer to~$v$ than~$w$.
\begin{center}
\includegraphics{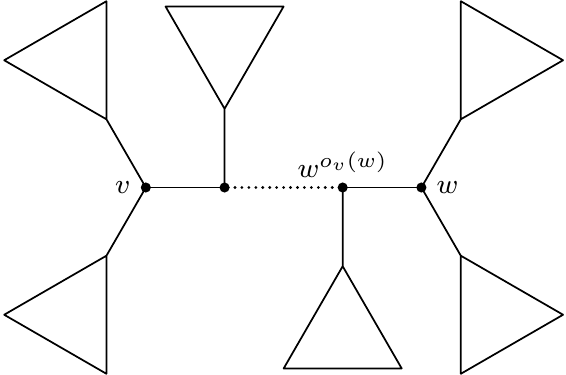}
%
%
%
%
%
%
%
%
%
\end{center}

\textup{(b)}
A $\fT$-tiling~$c$ is \defn{locally consistent towards $v$}
if, for all $w\in T$ and all directions $1\leq d\leq 3$ with $d \neq o_v(w)$,
we have
\begin{align*}
  c(w)_{\mi d} = \begin{cases}
                   c(w^d)_{\mo d} &\text{if } w^d \neq \triangle\,, \\
                   \triangle      &\text{otherwise}\,.
                 \end{cases}
\end{align*}

\textup{(c)}
A $\fT$-tiling~$c$ is \defn{globally consistent towards $v$}
if, for all vertices $w\in T$ and all directions $1\leq d\leq 3$ with $d \neq o_v(w)$,
we have
\begin{align*}
  c(w)_{\mi d} = \mtype n{(\fT,P)_{w^dw}}\,,
\end{align*}
where $(\fT,P)$ is the expansion of~$\fT$
by the set $P := \set{v\in T}{c(v)_{\mi0}=1}$ interpreted for~$X$.
\end{defi}

Of course, canonical tilings are globally consistent.
\begin{lem}\label{lem:tiling is such}
Let $\fT$~be a $\tau$-tree and $P \subseteq T$.
The $\fT$-tiling $t_{(\fT,P)}$ is
globally consistent towards each vertex $v\in T$.
\end{lem}
\begin{proof}
We have already seen in Lemma~\ref{lem:tile is such} that $t_{(\fT,P)}$ is a $\fT$-tiling.
Let $v \in T$. For global consistency, note that
\begin{align*}
  P &= \bigset{v\in T}{ (\fT,P,v) \models Xx} \\
    &= \bigset{v\in T}{\mtype n{(\fT,P)_{\setlr v}}=1} \\
    &= \bigset{v\in T}{t_{\fT,P}(v)_{\mi0}=1}\,,
\end{align*}
as desired.
\end{proof}

Finally, let us show that global consistency implies
local consistency.
\begin{lem}\label{lem:global implies local for tiling}
Let $\fT$ be a $\tau$-tree and $v\in T$.
Every $\fT$-tiling that is globally consistent towards~$v$
is locally consistent towards~$v$.
\end{lem}

\begin{proof}
Let $c$ be a $\fT$-tiling globally consistent towards~$v$
and let~$\fT'$ be the $\tauX$-expansion of~$\fT$ by the set
$P := \set{v\in T}{c(v)_{\mi0}=1}$.
Let $w\in T$ and $d \neq o_v(w)$ be given.
Without loss of generality, we may assume that $d=3$.
If $w^3 = \triangle$, then $c(w)_{\mi 3}$ is the type of
$\fT'_{\triangle v} = \triangle$.
Otherwise, let $u := w^3 \neq \triangle$.
Since $c$~is a $\fT$-tiling, $c(u)_{\mi 0}$~is either $0_{\fT_{\{u\}}}$
or~$1_{\fT_{\{u\}}}$. By definition of~$\fT'$ it follows that
$c(u)_{\mi 0} = \mtype{n}{\fT'_{\{u\}}}$. Consequently,
\begin{align*}
  c(w)_{\mi 3}
   = \mtype n{\fT'_{uw}}
  &= \mtype n{\fT'_{\setlr u}}
     \add^n \mtype n{\fT'_{u^1u}}
     \add^n \mtype n{\fT'_{u^2u}} \\
  &= c(u)_{\mi0} \add^n c(u)_{\mi1} \add^n c(u)_{\mi2} \\
  &= c(u)_{\mo 3}\,.
\end{align*}
\end{proof}

\section{Annotations}
\label{sect:annots}

Ideally we would like to annotate a given tree with one tiling for each stage
of the fixed-point induction. Since this is an infinite amount of data we have
to opt for something less\?: at each vertex of the tree we do not store the
full sequence of tiles for each stage, but only a shortened sequence
obtained by removing all duplicates. This is a finite amount of information
we can label the tree with.
The drawback of this method is that, by removing duplicates,
we lose synchronisation between the sequences from adjacent vertices.
Here are the formal definitions.

For a $\tau$-tree~$\fT$ and an ordinal~$\alpha$,
let $\fT^{\alpha} := \parlr{\fT,\app{\varphi^{\alpha}}{\fT}}$
be the $\tauX$-expansion of~$\fT$
by the $\alpha$th stage of the fixed-point induction.
Similarly, we set $\fT^\alpha_{vw,x} := (\fT^\alpha)_{vw,x}$
and $\fT^\alpha_{\{v\}} := (\fT^\alpha)_{\{v\}}$.

We extend the order $\subseteq$ on $X$-positive $n$-types
to tiles by requiring that $\subseteq$ holds component-wise.

\begin{defi}
\textup{(a)}
Let $\fL$ be a letter.
An \defn{$\fL$-history} is a strictly increasing sequence
$h=\parlr{h^0 \subsetneq \ldots \subsetneq h^m}$
of $\fL$-tiles such that
\begin{enumerate}
\item $h^0_{\mi0} = 0_{\fL}$ and
\item $h^{i+1}_{\mi0}=1_{\fL}$ iff $\varphi\in h^i_{\mo4}$, for $0\leq i<m$.
\end{enumerate}
The number~$m$ is the \defn{length} of the history, denoted $\len h$.

\textup{(b)}
Let $\fT$ be a $\tau$-tree and $v \in T$ a vertex.
\defn{The history} of~$\fT$ at~$v$
is the sequence~$h_\fT(v)$ of tiles $t_{\fT^{\alpha}}(v)$,
for all ordinals~$\alpha$, with duplicates removed.
\end{defi}

\begin{exa}\label{ex:paths}
For simplicity, we give an example of a fixed-point induction on
a path, instead of a tree, i.e., a tree where no vertex has a neighbour
in direction~$3$.
We consider the fixed-point of the formula $\varphi(X,x)$
stating that
\begin{align*}
  x^1 = \triangle \quad\text{or}\quad
  x^2 = \triangle \quad\text{or}\quad
  T_{x^1x} \subseteq X \quad\text{or}\quad
  T_{x^2x} \subseteq X\,.
\end{align*}
Figure~\ref{fig:annotated word} shows the histories of the
first $4$~elements of a finite path of length at least~$9$.
All further elements, except for the last two,
have the same history as the third and fourth elements.
Here, we assume that the edges are alternatingly labelled by $1$~and~$2$
and the tiles are drawn in the format
\begin{center}
\begin{tabular}{|c|c|c|c|}
\hline
$\mo4$&$\mo1$&$\mo2$&$\mo3$\\
\hline
$\mi0$&$\mi1$&$\mi2$&$\mi3$\\
\hline
\end{tabular}
\end{center}
where
\begin{itemize}\parskip=0pt\itemsep=0pt%
\item $\triangle$~denotes the type of the empty tree,
\item $\phi$~denotes any type containing $\phi$,
\item $\times$~denotes any type not containing $\phi$,
\item $\forall$~denotes any type containing the formula $\forall yXy$,
\item $\exists$~denotes any type containing $\exists yXy$,
  but not $\forall yXy$, and
\item $-$ denotes any type not containing the formula $\exists yXy$. 
\end{itemize}
\begin{figure}\centering
\includegraphics{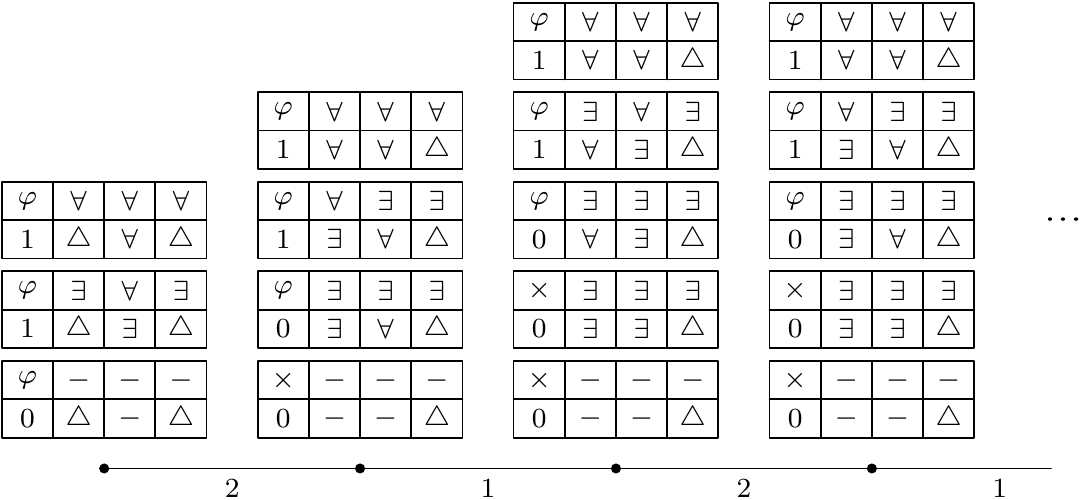}
\caption{Annotation for $\varphi(X,x)$\label{fig:annotated word}}
\end{figure}
\end{exa}

Of course, the history of~$\fT$ at~$v$ is indeed a history.
\begin{lem}\label{lem:history is such}
Let $\fT$~be a $\tau$-tree and $v \in T$ a vertex.
Then $\app{h_{\fT}}v$ is a $\fT_{\setlr v}$-history.
\end{lem}

\proof
Let $h := \app{h_{\fT}}v$.
We have already seen in Lemma~\ref{lem:tile is such}
that each $h^i$ is a $\fT_{\setlr v}$-tile.
The sequence is increasing,
because we are considering positive (hence monotone) types
and the sequence $\app{\varphi^{\alpha}}{\fT}$ is increasing.
It is strictly increasing because we have removed duplicates.
For ordinals~$\alpha$, let $\app k{\alpha}$ be the index
at which the $\alpha$th stage appears in~$h$, i.e.,
$h^{\app k{\alpha}} = t_{\fT^{\alpha}}(v)$.
As the sequence is increasing, so is~$k$.

Since $\varphi^0(\fT) = \emptyset$, we have
\begin{align*}
  h^0_{\mi0}
  = h^{k(0)}_{\mi0}
  = t_{\fT^0}(v)_{\mi0}
  = \mtype n{\fT^0_{\setlr v}}
  = 0_{\fL}\,.
\end{align*}
For $0\leq i<\len h$, let $\alpha$ be the minimal ordinal with $k(\alpha) = i+1$.
Then
\begin{align*}
  h^{i+1}_{\mi0}=1
  \quad\text{iff}\quad \mtype n{\fT^\alpha_{\setlr v}}=1
  \quad\text{iff}\quad v\in \varphi^\alpha(\fT)\,.
\end{align*}
Since elements enter the fixed point only at successor stages, we have
\begin{align*}
  v\in \varphi^\alpha(\fT)
  &\quad\text{iff}\quad v \in \varphi^{\beta+1}(\fT) \quad\text{for some } \beta < \alpha\,, \\
  &\quad\text{iff}\quad (\fT^\beta,v) \models \varphi \\
  &\quad\text{iff}\quad \varphi\in\mtype n{\fT^\beta,v} \subseteq
  h^i_{\mo4}\,.\rlap{\hbox to 144 pt{\hfill\qEd}}
\end{align*}

\noindent We would like to annotate each vertex~$v$ of a tree~$\fT$ by the sequence
$(t_{\fT^\alpha}(v))_\alpha$. To obtain a finite object, we have to remove
duplicates and, therefore, we work with the history $h_\fT(v)$ instead.
For each~$\alpha$, we would like to have an automaton that can recover the
tiling~$t_{\fT^\alpha}$ from~$h_\fT$.
In general, this is not possible.

For instance, in Example~\ref{ex:paths} the `real' tilings $t_{\fT^\alpha}(v)$
for a path~$\fT$ of even length are words of the form $ux^ny^nv$
where $y^nv$ is the `mirror image' of $ux^n$.
This language is not regular.

Hence, we use an approximation.
For each vertex~$v$, each index~$i$ of $h_\fT(v)$, and each direction~$d$,
we record the index~$j$ of $h_\fT(v^d)$ such that
$h_\fT(v^d)^j$ and $h_\fT(v)^i$ belong to the same ordinal~$\alpha$.
Of course, given~$i$, there are several choices of~$\alpha$
and, hence, of~$j$, so we lose information. It will turn out that
these two pieces of data,
the function~$h$ and the function $(v,i,d) \mapsto j$,
are sufficient for our purposes.

\begin{defi}
\textup{(a)}
An \defn{annotated tree} is a tuple $\parlr{\fT,h,s}$, where
\begin{enumerate}
\item $\fT$ is a $\tau$-tree,
\item $h$ is a mapping that assigns to each vertex $v\in T$ a $\fT_{\setlr v}$-history $h(v)$, and
\item $s$ is a mapping assigning a natural number $s(v,i,d)$
  to each vertex $v\in T$, each index $0\leq i\leq\len{\app hv}$, and every direction $1\leq d\leq 3$
  with $v^d \neq \triangle$.
\end{enumerate}
We call $h$~the \defn{history map} and $s$~the \defn{synchronisation} of the annotated tree.

\textup{(b)}
Let $(\fT,h,s)$ be an annotated tree.
For $v\in T$ and $0\leq i\leq\len{\app hv}$,
the \defn{section} at $v,i$
is the tiling~$c$ defined inductively as follows\?:
\begin{enumerate}
\item $\app cv:=\app hv^i$.
\item For $w\in T\smallsetminus\setlr v$, let $u:=w^{\app{o_v}w}$.
  We assume by induction that $\app cu$ is already defined.
  Let $j$ be the index such that $\app cu=\app hu^j$.
  Then we set $\app cw:=\app hw^{\app s{u,j,\app{o_v}w}}$.
\end{enumerate}
\end{defi}

\noindent Of course, not every annotated tree $(\fT,h,s)$
encodes the `real' fixed-point induction.
The next definition collects some necessary conditions.
\begin{defi}
Let $\cA = (\fT,h,s)$ be an annotated tree.

\textup{(a)}
$\cA$~is \defn{locally consistent}
if, for all vertices $v\in T$, indices $0\leq i\leq\len{\app hv}$, and directions $1\leq d\leq 3$
the following conditions are satisfied\?:
\begin{enumerate}
\item If $v^d=\triangle$,
  then $\app hv^i_{\mi d} = \triangle$.
\item Otherwise, $\app s{v,i,d}\leq\len{\app h{v^d}}$
  and $\app hv^i_{\mi d}=\app h{v^d}^{\app s{v,i,d}}_{\mo d}$.
\end{enumerate}

\textup{(b)}
$\cA$~is \defn{globally consistent}
if it is locally consistent
and if, for all $v,i$ as above, the section at $v,i$
is globally consistent towards~$v$.
\end{defi}

\begin{lem}\label{lem: sections are locally consistent}
Let $(\fT,h,s)$ be a locally consistent annotated tree.
Every section~$c$ at some $v,i$ is locally consistent towards~$v$.
\end{lem}
\begin{proof}
Let $v,w \in T$ be distinct vertices, $d \neq o_v(w)$, and let $i$~be the index
such that $c(w) = h(w)^i$.
Then we have $w^d=\triangle$ and $c(w)_{\mi d}=h(w)^i_{\mi d}=\triangle$, or
\begin{align*}
  c(w)_{\mi d} = h(w)^i_{\mi d} = h(w^d)^{s(w,i,d)}_{\mo d} = c(w^d)_{\mo d}\,.
\end{align*}
\end{proof}

We have not yet defined the `real annotation' of a tree.
In fact, due to the choices involved in defining
the synchronisation there are several possible `real' annotations.
We obtain them by fixing an ordinal~$\beta$ and selecting
that synchronisation that
selects from among all possible choices
the stage that is closest to~$\beta$.

\begin{defi}\label{def: betasyndef}
Let $\fT$ be a $\tau$-tree and $\beta<\omega$.
We denote by $\cA_\beta(\fT)$
the annotated tree $\parlr{\fT,h_{\fT},s}$
where the synchronisation~$s$ is defined as follows.
For $v\in T$, $0\leq i\leq\len{\app{h_{\fT}}v}$,
and $1\leq d\leq 3$ with $w:=v^d \neq \triangle$,
we define $\app s{v,i,d}$ such that
\begin{align*}
  \app{h_{\fT}}w^{\app s{v,i,d}}=t_{\fT^{\alpha}}(w)\,,
\end{align*}
where the ordinal~$\alpha$ is chosen as follows\?:
\begin{enumerate}
\item if $\app{h_{\fT}}v^i = t_{\fT^{\beta}}(v)$, then $\alpha=\beta$,
\item if $\app{h_{\fT}}v^i\subsetneq t_{\fT^{\beta}}(v)$,
  then $\alpha\leq\beta$ is maximal
  such that $t_{\fT^{\alpha}}(w)_{\mo d}=\app{h_{\fT}}v^i_{\mi d}$, and
\item if $\app{h_{\fT}}v^i\supsetneq t_{\fT^{\beta}}(v)$,
  then $\alpha\geq\beta$ is minimal
  such that $t_{\fT^{\alpha}}(w)_{\mo d}=\app{h_{\fT}}v^i_{\mi d}$.
\end{enumerate}
\end{defi}

We start with a technical lemma containing a monotonicity property
for the sections of an annotation.
\begin{lem}\label{lem:sections monotone}
Let $\fT$~be a tree with vertices $v,w \in T$,
let $c$~be the section of $\cA_\beta(\fT)$ at $v,i$,
and set $d := o_v(w)$.
\begin{enumerate}
\item[\normalfont(a)] If $c(w) = t_{\fT^\beta}(w)$,
  then $c(u) = t_{\fT^\beta}(u)$, for all $u \in T_{ww^d}$.
\item[\normalfont(b)] Let $\alpha < \beta$. If $t_{\fT^\alpha}(w) \subseteq c(w)$, then
  $t_{\fT^\alpha}(u) \subseteq c(u)$, for all $u \in T_{ww^d}$.
\end{enumerate}
(Here we set $T_{ww^4} := T$.)
\end{lem}
\begin{proof}
We only prove the claims for $w \neq v$. The argument for $w = v$ is similar.
We prove both claims by induction on the distance between $u \in T_{ww^d}$
and~$v$. The claims are immediate for $u=w$.

For the inductive step assume that the claims hold for~$u$
and let $u'$ be a neighbour of~$u$ which is further away from~$v$ than~$u$,
so that $u = (u')^{o_v(u')}$.
It follows that $u'=u^{d'}$ for some $d' \neq o_v(u)$.
Let $i$~be the index such that $c(u) = h_\fT(u)^i$.
By definition of $c$~and~$s$, respectively, we have
\begin{align*}
  c(u') = h_{\fT}(u')^{s(u,i,d')} = t_{\fT^{\alpha'}}(u')\,,
\end{align*}
for some ordinal~$\alpha'$.

For (a), using the inductive hypothesis,
we have $h_\fT(u)^i = c(u) = t_{\fT^\beta}(u)$,
which implies that $\alpha' = \beta$.
Hence, $c(u') = t_{\fT^\beta}(u')$.

Similarly, for (b), we have $h_\fT(u)^i = c(u) \supseteq t_{\fT^\alpha}(u)$,
which implies that $\alpha' \geq \alpha$.
Hence, $c(u') \supseteq t_{\fT^\alpha}(u')$.
\end{proof}

Let us also show that $\cA_\beta(\fT)$ is always globally consistent.
\begin{lem}\label{lem:annotated tree is such}
For all $\tau$-trees $\fT$ and all $\beta<\omega$,
$\cA_\beta(\fT)$ is a globally consistent annotated tree.
\end{lem}

\proof
We have seen in Lemma~\ref{lem:history is such} that $h_\fT(v)$ is a
$\fT_{\{v\}}$-history.
Hence, $\cA_\beta(\fT)$ is an annotated tree.
For local consistency, fix $v,i,d$ and let $\alpha$~be the ordinal
from the definition of~$s$
at~$v$ in $\cA_\beta(\fT)$ (cf.~Definition~\ref{def: betasyndef}).
Then
\begin{align*}
  h_\fT(v^d)^{s(v,i,d)}_{\mo d} = t_{\fT^\alpha}(v^d)_{\mo d} = h_\fT(v)^i_{\mi d}\,,
\end{align*}
as desired.

It remains to prove global consistency.
Fix a vertex $v\in T$ and an index $0\leq i\leq\len{\app{h_{\fT}}v}$,
and let $c$~be the section at~$v,i$.
For $w\in T$, let $\app{\alpha}w$ be the ordinal closest to~$\beta$
such that $\app cw = t_{\fT^{\app{\alpha}w}}(w)$.
Let $\fT'$ be the expansion of $\fT$
by the set $P := \set{w\in T}{\app cw_{\mi0}=1}$.
We need to show that
$\app cw_{\mi d}=\mtype n{\fT'_{w^dw}}$,
for all $w\in T$ and $d \neq \app{o_v}w$.
(Here, $\fT'_{v^4v}:=\fT'$.)
By local consistency of~$c$
(which holds by Lemma~\ref{lem: sections are locally consistent}),
it is sufficient to show that
$\app cw_{\mo d}=\mtype n{\fT'_{ww^d}}$,
for all $w\in T$ and $d:=\app{o_v}w$.
We do this by induction on the distance between $\alpha(w)$ and $\beta$.

First, suppose that $\alpha(w) = \beta$.
Then Lemma~\ref{lem:sections monotone}\,(a) implies that
$c(u) = t_{\fT^\beta}(u)$, for all $u \in T_{ww^d}$.
Consequently $\fT'_{ww^d} = \fT^\beta_{ww^d}$,
and hence $\mtype n{\fT'_{ww^d}} = t_{\fT^\beta}(w)_{\mo d} = c(w)_{\mo d}$, as desired.

It remains to consider the case that $\alpha(w) \neq \beta$.
By symmetry, we may assume that $\alpha(w) > \beta$.
Let $\fS$~be the maximal subtree of $\fT'_{ww^d}$
that contains the vertex~$w$ and
such that $\alpha(u) = \alpha(w)$ for all $u\in S$.
Let $\parlr{x_1,d_1},\parlr{x_2,d_2},\dots$ be the finite or infinite list
of all pairs $(x,d)$ such that $x \in S$ and $x^d \in T \smallsetminus S$.
Let $y_k := x_k^{d_k}$ be the missing neighbour.
Note that, by definition of~$s$, $\alpha(y_k)$ is the minimal ordinal~$\alpha$
such that
$\mtype n{\fT^\alpha_{y_kx_k}} = \mtype n{\fT^{\alpha(x_k)}_{y_kx_k}}$.
Hence, $\alpha(y_k) \leq \alpha(x_k)$ and it follows that
$\alpha(x_k) = \alpha(w)$ and $\alpha(y_k) < \alpha(w)$.
By local consistency and the inductive hypothesis, we have
\begin{align*}
  c(x_k)_{\mi d_k} = c(y_k)_{\mo d_k} = \mtype n{\fT'_{y_kx_k}}\,,
\end{align*}
while, by definition of~$S$, we have
\begin{align*}
  c(x_k)_{\mi d_k} = t_{\fT^{\alpha(w)}}(x_k)_{\mi d_k} = \mtype n{\fT^{\alpha(w)}_{y_kx_k}}\,.
\end{align*}
It follows that
$\mtype n{\fT'_{y_kx_k}}=\mtype n{\fT^{\app{\alpha}w}_{y_kx_k}}$
for each index~$k$.
As the subtrees of
$\fT'_{ww^d}$ and $\fT^{\app{\alpha}w}_{ww^d}$ induced by~$S$ agree,
we can use Proposition~\ref{prop:Feferman-Vaught II}
to deduce that
\[
  c(w)_{\mo d} = \mtype n{\fT^{\alpha(w)}_{ww^d}} = \mtype n{\fT'_{ww^d}}\,.\eqno{\qEd}
\]

\section{Ranks}
\label{sect:ranks}

It remains to compute the length of the fixed-point iteration from
a given annotated tree.
The goal essentially is to obtain an estimate for the stage of a designated
element of the fixed point\?; this estimate is extracted from an annotation
in terms of the weight of an accepting run of a weighted
automaton which checks consistency of the annotation.
The appropriate kind of weighted automata for this purpose will be
presented in the next section.

\begin{defi}
\textup{(a)}
Let $(\fT,h,s)$ be an annotated tree and $v \in T$ a node.
We say that there is a \defn{jump} at~$v$
if there is some index~$i$ such that
$h(v)^i_{\mi0} = 0$ and $h(v)^{i+1}_{\mi0} = 1$.
Observe that this value of $i$~is uniquely determined.
We call the jump a \defn{base jump} if $i=0$.

\textup{(b)}
Suppose that there is a jump at~$v$ that is not a base jump.
We say that this jump \defn{depends} on another jump at a node~$w$
if $c(w)_{\mi0}=0$ where $c$~is the section at $v$ and $i-1$.
The \defn{rank} of a jump is the minimal number of jumps
on any dependency chain from this jump to some base jump.

\textup{(c)}
An annotated tree $(\fT,h,s)$ is \defn{jump-consistent,}
if the set of vertices with a jump equals $\app{\varphi^{\infty}}{\fT}$.
\end{defi}

The notion of \emph{dependency} in~(b) may warrant some comment, because
the terminology could easily be misunderstood. What the criterion is meant to
capture is not that there must be a (causal or temporal) dependence of the
appearance of~$v$ in the fixed point on the (prior) appearance of~$w$\?;
rather, it says that such a dependence cannot be ruled out.
At least any~$w$ that $v$~does \emph{not} depend on in the sense of the definition can
have had no influence on the appearance of~$v$.
In this sense our dependency relation provides a generous upper bound
on any intuitive `real' dependency\?:
it may be useful to think of $w$~as a \emph{potential} trigger for~$v$.

Let us compare the rank of a jump with the stage of the corresponding vertex
in the fixed point (the stage at which the vertex enters the fixed point).
First, we show that in every annotation the latter bounds the former.
\begin{lem}\label{lem:stage bounds min rank}
Let $\cA=\parlr{\fT,h,s}$
be a globally consistent and jump-consistent annotated tree,
let $v$~be a vertex with a jump in~$\cA$, and $\alpha < \omega$.
If $v\in\app{\varphi^{\alpha}}{\fT}$,
then the rank of~$v$ is at most~$\alpha$.
\end{lem}

\proof
We proceed by induction on~$\alpha$.
For $\alpha = 0$ there is nothing to do since
$\varphi^0(\fT) = \emptyset$.
Hence, we may assume that $\alpha > 0$ and that the claim already holds for smaller
ranks.
Let $i$ be the index such that $\app hv^i_{\mi0}=0$ and $\app hv^{i+1}_{\mi0}=1$.

If $i=0$, then there is a base jump at~$v$ and its rank is $1 \leq \alpha$.

For $i>0$, let $c$ be the section at $v,i-1$
and let $P:=\set{w\in T}{\app cw_{\mi0}=1}$.
From $\app hv^i_{\mi0}=0$ we conclude that $\varphi \notin \app cv_{\mo4}$.
As $c$ is globally consistent, it follows that $\varphi \notin \mtype n{\fT,P,v}$.
On the other hand, $\varphi\in\mtype n{\fT,\app{\varphi^{\alpha-1}}{\fT},v}$.
By monotonicity,
there must be some vertex $w \in \varphi^{\alpha-1}(\fT) \smallsetminus P$,
which, by jump-consis\-tency, has a jump.
As $w \notin P$, we have $\app cw_{\mi0}=0$. Consequently, $v$~depends on~$w$.
By inductive hypothesis, $w$~has rank at most $\alpha-1$.
Therefore, $v$~has rank at most~$\alpha$.
\qed

Some form of converse is true for annotations of the form $\cA_\beta(\fT)$.
\begin{lem}\label{lem:min rank bounds stage}
Let $\fT$ be a $\tau$-tree, $v \in \varphi^\infty(\fT)$, and $\alpha<\beta<\omega$.
If the rank of $v$ in $\cA_\beta(\fT)$ is at most $\alpha$,
then $v\in\app{\varphi^{\alpha}}{\fT}$.
\end{lem}

\proof
We proceed by induction on $\alpha$.
As all ranks are positive, $\alpha>0$.
Let $i$ be the index such that $\app{h_{\fT}}v^i_{\mi0}=0$ and $\app{h_{\fT}}v^{i+1}_{\mi0}=1$.
If $i=0$,
then $\varphi\in\app{h_{\fT}}v^0_{\mo4}=t_{\fT^0}(v)$.
Therefore, $v\in\app{\varphi^1}{\fT}\subseteq\app{\varphi^{\alpha}}{\fT}$ and we are done.
Hence we may assume that $i>0$, i.e., the jump at~$v$ is not a base jump.

Let $c$~be the section at $v,i-1$.
As the jump at~$v$ is not a base jump and its rank is finite,
there is some vertex~$w$ with a jump rank at most $\alpha-1$
such that $v$~depends on~$w$.
By the inductive hypothesis, we have $w\in \varphi^{\alpha-1}(\fT)$.
Consequently, $t_{\fT^{\alpha-1}}(w)_{\mi0}=1$.
On the other hand, we have $c(w)_{\mi0} = 0$ by choice of~$w$.
Therefore, $c(w) \subsetneq t_{\fT^{\alpha-1}}(w)$.
By Lemma~\ref{lem:sections monotone}, this implies that
\begin{align*}
  \app{h_{{\fT}}}v^{i-1} = c(v) \subsetneq t_{\fT^{\alpha-1}}(v)\,.
\end{align*}
It follows that $\alpha'<\alpha-1<\beta<\omega$
for any~$\alpha'$ such that $\app{h_{\fT}}v^{i-1}=\app{t_{\fT^{\alpha'}}}v$.
Therefore, there exists a maximal such ordinal~$\alpha'$ and,
moreover, $\alpha'+2\leq\alpha$.

As $i$~is maximal such that $\app{h_{\fT}}v^i_{\mi0}=0$,
it follows that $i-1$ is maximal
such that $\varphi\notin\app{h_{\fT}}v^{i-1}_{\mo4}$.
Accordingly, $\alpha'$~is maximal such that
$\varphi\notin\app{t_{\fT^{\alpha'}}}v_{\mo4}$.
Thus, $\varphi\in\app{t_{\fT^{\alpha'+1}}}v_{\mo4}$
and $\app{t_{\fT^{\alpha'+2}}}v_{\mi0}=1$.
It follows that $v\in\app{\varphi^{\alpha'+2}}{\fT}\subseteq\app{\varphi^{\alpha}}{\fT}$.
\qed

It follows that boundedness of the fixed-point iteration
is equivalent to the existence of a finite bound on the ranks of all annotations.
\begin{defi}
A \defn{proposal} is a tuple $\parlr{\fT,h,s,v}$,
where $\parlr{\fT,h,s}$
is a globally consistent and jump consistent annotated tree
and $v\in\app{\varphi^{\infty}}{\fT}$.
The \defn{rank} of such a proposal is the rank of the jump at~$v$ in $\parlr{\fT,h,s}$.
\end{defi}

\begin{prop}\label{prop: boundeness and ranks}
A formula~$\varphi$ is bounded over the class of all ternary trees
if, and only if, there is some number $N < \omega$
such that the rank of each proposal is at most~$N$.
\end{prop}

\proof
$(\Leftarrow)$ Suppose there exists a bound $N < \omega$ on the ranks of proposals.
Let $\fT$ be some ternary tree, and $v\in\app{\varphi^{\infty}}{\fT}$.
By Lemma~\ref{lem:annotated tree is such},
$\parlr{\cA_{N+1}(\fT),v}$ is a proposal.
By choice of~$N$, the rank of~$v$ is at most~$N$.
Hence, Lemma~\ref{lem:min rank bounds stage} implies that $v\in\app{\varphi^N}{\fT}$.
As $v$ was arbitrary, it follows that $\app{\varphi^N}{\fT}=\app{\varphi^{\infty}}{\fT}$.

$(\Rightarrow)$ Suppose that $\varphi$~is bounded by some number $N < \omega$.
Let $\parlr{\fT,h,s,v}$ be an arbitrary proposal.
Then $v \in \varphi^\infty(\fT) = \varphi^N(\fT)$
and Lemma~\ref{lem:stage bounds min rank} implies
that the rank of the proposal is at most~$N$.
\qed

\section{Weighted automata}
\label{sect:automata}
\label{sect:end I}

In order to decide the boundedness problem for $\MSO$
we reduce it to the so-called \emph{limitedness problem}
for a certain kind of weighted automaton.
These automata have \defn{$\Sigma$-labelled directed trees} as inputs.
Such a tree is a triple $\parlr{T,E,\lambda}$
where $\lambda:T\to\Sigma$ is a labelling of~$T$
and $\parlr{T,E}$ is a directed tree
(meaning that $E\cap E^{-1}=\emptyset$,
$\parlr{T,E\cup E^{-1}}$ is a tree structure,
and there is some $r\in T$ called the \defn{root} of the tree
such that $r \mathrel{E^*} t$ for all $t\in T$).

\begin{defi}
\textup{(a)}
A \emph{weighted parity automaton} $\cA = (Q,\Sigma,\Delta,I,\Omega,w)$ consists
of a finite \emph{state space}~$Q$, a finite \emph{input alphabet}~$\Sigma$,
a set $I \subseteq Q$ of \emph{initial states,}
a finite \emph{transition relation}
$\Delta \subseteq \Sigma\times\omega^Q\times Q$,
a \emph{priority function} $\Omega : Q \to \omega$,
and a \emph{weight function} $w: \Delta\to\omega$.

A weighted parity automaton~$\cA$ takes as input 
$\Sigma$-labelled directed trees $\parlr{T,E,\lambda}$.
Let $\pi_3 : \Sigma\times\omega^Q \times Q \to Q$ be the projection
to the third component.
A \emph{run} of $\cA$ on this tree is a mapping $\varrho:T\to\Delta$
satisfying, for all vertices $v \in T$, the following condition\?:
\begin{align*}
  \varrho(v) = (c,f,q)
  \quad\text{implies}\quad
  &c = \lambda(v) \text{ and } f \text{ is the function mapping } p \in Q
  \text{ to} \\
  &\text{the number of children } u \text{ of } v \text{ with }
   \pi_3(\varrho(u)) = p\,.
\end{align*}

A run~$\varrho$ is \emph{accepting,} if
\begin{itemize}
\item $\pi_3(\varrho(r)) \in I$ for the root~$r$ of~$\parlr{T,E,\lambda}$ and
\item for every branch $\beta$ of $(T,E)$, the limit
  \begin{align*}
    \liminf_{v \in \beta} \Omega(\pi_3(\varrho(v))) \quad\text{is even.}
  \end{align*}
\end{itemize}
The \emph{language} $L(\cA)$ recognised by~$\cA$
is the set of all $\Sigma$-labelled directed trees~$\parlr{T,E,\lambda}$
on which there is an accepting run of~$\cA$.

For a run~$\varrho$ on some tree~$\parlr{T,E,\lambda}$
and a branch~$\beta$ of the tree, we set
\begin{align*}
  w_\cA(\varrho,\beta) := \sum_{v\in\beta} w(\varrho(v))
  \quad\text{and}\quad
  w_\cA(\varrho) := \sup_\beta w_\cA(\varrho,\beta)\,,
\end{align*}
with values in $\omega \cup \{ \infty \}$.

The associated \emph{cost function} $w_\cA$ maps 
$\parlr{T,E,\lambda} \in L(\cA)$ to
the minimum of $w_\cA(\varrho)$ 
taken over all accepting runs~$\varrho$ on~$\parlr{T,E,\lambda}$.
If $(T,E,\lambda) \notin L(\cA)$, $w_\cA$ returns~$\infty$.

\textup{(b)}
We say that the automaton~$\cA$ is \emph{limited,}
if there is some bound $N < \omega$
such that $w_\cA(T,E,\lambda) \leq N$ for all $(T,E,\lambda) \in L(\cA)$.
We say that $\cA$~is \emph{limited in the finite,}
if there is a bound $N < \omega$ such that
$w_\cA(T,E,\lambda) \leq N$ for all \emph{finite} $(T,E,\lambda) \in L(\cA)$.
\end{defi}
Note that, if we only consider finite trees as input, we can omit the priority
function~$\Omega$ from the automaton. Weighted automata as defined above
are a special case of so-called \emph{cost tree automata}
introduced in~\cite{ColcombetLoeding08}.
In that paper it is shown than the limitedness problem for cost tree automata
over finite trees is decidable.
Hence, the following is a direct
consequence of~\cite{ColcombetLoeding08}.
\begin{thm}[Colcombet and L\"oding]\label{thm:finite limitedness decidable}
It is decidable whether a weighted parity automaton~$\cA$ is limited in the
finite.
\qed\end{thm}
Colcombet and L\"oding have also announced a decidability result for the
general limitedness problem, but this result has not been published yet.
\begin{thm}[Colcombet and L\"oding]\label{thm:limitedness decidable}
It is decidable whether a weighted parity automaton~$\cA$ is limited.
\qed\end{thm}

Although the proof is still not published, its key arguments appear
in \cite{VandenBoom12,Colcombet13}.
The following sketch of how they fit together was communicated to the
authors by Colcombet and L\"oding.

A \emph{cost function} $f : \cT \to \omega \cup \{\infty\}$
associates with every tree a natural number or $\infty$.
We say that such a cost function~$f$ is \emph{dominated} by~$g$ if
$f$~is bounded over every subset $X \subseteq \cT$ over which~$g$ is bounded.
We denote this domination relation by $f \preceq g$.

We can state Theorem~\ref{thm:limitedness decidable} in terms of the
domination relation as follows.
Let $\cA$~be a weighted automaton and let $L$~be the language defined by~$\cA$
if we consider it as an ordinary parity automaton without weight function.
Let $f$~be the cost function $w_\cA$ associated with~$\cA$ and let $g$~be the
cost function that maps every tree in~$L$ to~$0$ and every other tree
to~$\infty$.
Then Theorem~\ref{thm:limitedness decidable} states that it is decidable
whether $f \preceq g$.

We would like to reduce this statement to Corollaire~8.11 of~\cite{Colcombet13},
which states -- in the terminology of~\cite{Colcombet13} --
that the domination relation $f \preceq g$ between cost functions
$f$~and~$g$ is decidable, provided that $f$~is given by a nondeterministic
$S$-Muller automaton and $g$~is given by a nondeterministic $B$-Muller
automaton.
The function~$g$ from above is given by a parity automaton without weight
function.
Such an automaton can trivially be converted into a $B$-Muller automaton.
Hence, to complete the proof it remains to find an $S$-Muller automaton
recognising~$f$.
This can be done in the same way as in the proof of Theorem~4.28
of~\cite{VandenBoom12}, where the author shows how to transform an
alternating $B$-B\"uchi automaton into a nondeterministic one.
This proof uses game-theoretic techniques.
One key argument is the fact that the games, which correspond to the
automata in question, are positionally determined.
To adapt the proof to our case, one needs positional determinacy for games
whose winning condition is a disjunction between an unboundedness condition
and a parity condition.
This can be shown as in Proposition~7.14 of~\cite{Colcombet13},
which treats winning conditions consisting of a conjunction of a boundedness
condition and a Rabin condition.
One further step of adaptation consists in the construction of a so-called
`history-deterministic' automaton that checks whether a given positional
strategy is winning.
For finite words, the underlying translation of nondeterministic automata into
history-deterministic ones can be found in an unpublished note
(cf.~Lemma~58 of~\cite{Colcombet09b}) on the author's web-page.

\medskip
Using the results of the previous sections we can reduce
the boundedness problem for $\MSO$ on ternary trees
to Theorem~\ref{thm:limitedness decidable}.
To do so, we construct a weighted automaton computing the rank of a proposal
$(\fT,h,s,v)$.
In order to use $\parlr{\fT,h,s,v}$ as input for a tree automaton,
we encode it as a labelled directed tree with root~$v$.
The labelling contains information about the unary predicates in~$\tau$,
the histories, and the synchronisation.
As there is only a finite number of types,
there is a uniform bound on the length of histories
and we only need finitely many labels.

First we show that the set of all proposals is regular.
\begin{lem}\label{lem: proposals are regular}
Given a formula~$\varphi$,
we can effectively construct a parity automaton~$\cA$ recognising
the set of all proposals for~$\varphi$.
\end{lem}

\proof
Let $n$ be the quantifier rank of $\varphi$.
It is sufficient to show that the set of proposals can be defined in $\MSO$.
Being a locally consistent annotated tree can be expressed even in~$\FO$
since it is a purely local property.

For global consistency, note that we can encode a section~$c$
by a tuple of unary predicates~$\bar C$
(the precise number depends on the maximal length of a history)
such that there is an $\FO$-formula $\vartheta_i(v)$
stating that $\bar C$ encodes the section at $v,i$.
Thus, the section at $v,i$ is $\MSO$-definable and
the corresponding tiling is $\MSO$-interpretable.
In this tiling it is of course possible by means of $\MSO$
to determine the $\MSO$-type (of quantifier rank at most~$n$) of a subtree.
Consequently, we can express the global consistency of the tiling
and, hence, also the global consistency of the annotated tree.

As the set of jumps can be inferred from the tree labelling,
it is easy to check whether there is a jump at the root ($v \in \phi^\infty$).

It remains to consider jump-consistency,
that is, it remains to define $\app{\varphi^{\infty}}{\fT}$
(where $\fT$ is the first component of the prospective proposal).
As $\varphi$ is positive in~$X$, this can be achieved by
\[
  \psi(x) := \forall X[\forall y(\varphi(X,y) \to Xy) \to Xx] \,.\eqno{\qEd}
\]

\begin{lem}\label{lem: proposals and automata}
Given a formula~$\varphi$,
we can effectively construct a weighted parity automaton~$\cA$ such that
\begin{enumerate}
\item $L(\cA)$ is the set of proposals of finite rank\?;
\item if $P$ is a proposal and $r<\omega$~its rank,
  then $\frac{1}{2} \log r \leq w_\cA(P) \leq r$.
\end{enumerate}
\end{lem}

\begin{proof}
Let $n$ be the quantifier rank of~$\varphi$
and let $\cA_1$~be the automaton from Lemma~\ref{lem: proposals are regular}.
We will construct the desired automaton~$\cA$
as a product of~$\cA_1$ and a weighted parity automaton~$\cA_2$,
where the weight function of~$\cA$ is that of $\cA_2$.

Recall that the rank of a proposal $P=\parlr{\fA,h,s,v}$
is the minimal number of jumps
on a dependency chain from the jump at~$v$ to some base jump.
By this minimality condition
we can restrict our attention to chains without cycles.
Each dependency in the chain, say from~$u$ on~$u'$,
corresponds to a path in the section at $u,i$ for a suitable~$i$.
By minimality again, we only need to consider
pairwise disjoint paths, one for each dependency in the chain.
(If two paths intersected, we could form a new path witnessing the dependency
of some former jump in the chain to a latter one. This could be used to shorten
the dependency chain.)
These paths can be concatenated to form a single path in the annotated tree.
For a dependency path~$p$ and a tree node~$u$,
we say that $u$~is \defn{active}
if there is at least one jump on~$p$ in the subtree rooted at~$u$.

Since the tree is ternary, we can encode dependency paths by a tuple of unary predicates.
We first construct a weighted parity automaton~$\cA_3$ that takes as input
a proposal together with such a path.
It checks that the path follows the synchronisation (except for the jumps),
and that it is indeed a single path.
Furthermore, $\cA_3$~is such that from its state at a node~$u$
one can deduce whether $u$~is active.
We define the weight function of $\cA_3$ such
that all transitions have weight $0$~or~$1$, where we assign
a weight of~$1$ if at least two children of the current node are active
or if there is a jump at the current node.

For a dependency path~$p$ in a proposal~$P$,
let us compare its number~$r$ of jumps
with the weight computed by~$\cA_3$.
Let $\varrho$~be any accepting run of~$\cA_3$
on the input $\parlr{P,p}$.
We claim that $\frac12\log r\leq\app{w_{\cA_3}}{\varrho}\leq r$.

For the second inequality,
let $\beta$~be a branch of $P$ which realizes the maximum for~$\varrho$,
that is, $\app{w_{\cA_3}}{\varrho}=\app{w_{\cA_3}}{\varrho,\beta}$.
With each node $u\in\beta$ such that $\app{w_{\cA_3}}{\app{\varrho}u}=1$
we associate a jump in~$p$ as follows\?:
if there is a jump at~$u$, we just take this jump.
Otherwise, $u$~has at least two active children,
so it has at least one active child not in~$\beta$.
We take some jump from the subtree rooted at that child.
It is clear that, for different $u\in\beta$, we have chosen different jumps.
Hence, $r\geq\app{w_{\cA_3}}{\varrho,\beta}=\app{w_{\cA_3}}{\varrho}$.

For the other inequality, we construct a branch~$\beta$ as follows\?:
the branch starts at the root
and, whenever we have constructed~$\beta$ up to some node~$u$
which is not a leaf, we extend~$\beta$ with a child~$u'$ of~$u$
such that the number of jumps on~$p$ in the subtree rooted at~$u'$
is at least as large as the respective number for any other child of~$u$.
Let us trace this number along~$\beta$. Initially, it is~$r$.
It never increases and, whenever it decreases,
the respective transition has weight~$1$
by construction of~$\cA_3$.
As we always descend into the fattest subtree,
the number cannot decrease indefinitely\?:
if it is~$m$ for some node, it is at least $\frac{m-1}3$ for its child
(recall that the original undirected tree is ternary,
so the directed tree has branching at most~$3$,
and even at most~$2$ apart from the root).
A very rough analysis gives that, if $r\geq 4^k$,
then at least~$k$ decreasing steps occur on~$\beta$.
Hence,
$\app{w_{\cA_3}}{\varrho}\geq\app{w_{\cA_3}}{\varrho,\beta}\geq\frac12\log r$.

Finally, we obtain the desired automaton~$\cA_2$ from~$\cA_3$
by nondeterministically guessing the extra component~$p$.
To see that the product automaton~$\cA$ has the claimed properties,
let~$P$ be an input for~$\cA$.
If $P$ is a proposal of finite rank,
then it is in particular a proposal. Hence, $\cA_1$ accepts~$P$.
As the rank is finite, there is some dependency path $p$ for~$P$.
Therefore, $\cA_3$ accepts $\parlr{P,p}$ and $\cA_2$ accepts~$P$.
Consequently, also $\cA$ accepts~$P$.
For the converse, assume that $\cA$ accepts~$P$.
Then $P$ is a proposal since $\cA_1$ accepts~$P$.
Furthermore, there is some~$p$ such that $\cA_3$ accepts $\parlr{P,p}$.
Thus, $p$~is a dependency path in~$P$ and $P$ has finite rank.
Now, assume that $P$ is a proposal of rank~$r$
and let~$p$ be a dependency path in~$P$ with $r'$~jumps.
Let $\varrho$ be the accepting run of $\cA_3$ on $\parlr{P,p}$
and let $\varrho'$ be the corresponding accepting run of $\cA$ on $P$.
For each accepting run of $\cA$ on $P$
there is such a $p$ by construction of~$\cA$.
If $p$ is such that $\app{w_{\cA}}{\varrho'}$ is minimal,
then we can deduce
$\frac12\log r\leq\frac12\log r'\leq\app{w_{\cA_3}}{\varrho}
  =\app{w_{\cA}}{\varrho'}=\app{w_{\cA}}P$.
If, on the other hand, $p$~is such that $r'$ is minimal,
we obtain
$\app{w_{\cA}}P\leq\app{w_{\cA}}{\varrho'}=\app{w_{\cA_3}}{\varrho}
  \leq r'=r$.
\end{proof}

Combining our results we obtain a proof of the following theorem.

\begin{thm}\label{thm:boundedness for ternary trees}
The boundedness problem for $\MSO$ on the class
of all ternary trees is decidable.
\end{thm}
\begin{proof}
Given an $\MSO$-formula~$\varphi$, we construct
the weighted automaton~$\cA$ from Lemma~\ref{lem: proposals and automata}.
By Proposition~\ref{prop: boundeness and ranks}, it follows
that $\varphi$~is bounded if, and only if, $\cA$ is limited.
The latter we can decide with the help of
Theorem~\ref{thm:limitedness decidable}.
\end{proof}

\bigskip
\section*{Part II. Ramifications}

The boundedness problem has long been of interest
both in classical model theory and in the study of
the algorithmic properties of various fragments,
which in turn is partly motivated by applications
in computer science. The seminal result in the classical
model theory of the boundedness problem 
is the theorem of Barwise and Moschovakis~\cite{BarwiseMoschovakis78}
(see Theorem~\ref{thm:BMthm} below)\?;
the main interest in boundedness as a decision problem,
on the other hand, stems from an interest in \textsc{Datalog}
query optimisation as highlighted in the
first positive and negative results in
\cite{GaifmanMaSaVa93,HillebrandEtAl95}.
In both contexts, the natural emphasis
was on (not necessarily monadic) monotone inductions based on
first-order formulae or formulae in specific fragments of first-order
logic. Even in the study of rather weak fragments of first-order logic,
undecidability of the boundedness problem turned out to be the rule,
decidability the rare exception.

In this second part we link our new results to the wider setting
of the boundedness problem. After a short introduction to
this wider setting, we employ some rather more traditional tools
from model theory, like transfer results and interpretations,
to generalise the technical core results of Part~I and to
reap a number of further specific decidability results.
Some of these answer key open questions raised in the
more traditional setting, concerning, for instance, decidability
of boundedness for the guarded fragment or for the modal $\mu$-calculus.

To this end, we first review the shift in perspective
from boundedness for syntactically restricted fragments of $\FO$
to boundedness over restricted classes of structures\?;
a shift that was first explicitly proposed in \cite{KOS}
where boundedness for otherwise unconstrained monadic $\FO$
is treated over the class of acyclic structures.
The class $\cA$ of acyclic structures consists of those
structures whose Gaifman graph is acyclic.

\begin{thm*}[\cite{KOS}]
The boundedness problem for monadic least fixed points
of arbitrary $X$-positive $\FO$-formulae over the class of all
acyclic relational structures, $\BDDm(\FO,\cA)$, is decidable.
\qed\end{thm*}

The interest here was due to the observation
that reductions to settings involving tree-like structures seem to
be a common theme in most decidability results for boundedness.
On the other hand, availability of grid-like structures can
be widely used to show undecidability of boundedness issues via
reductions from tilings \cite{KOunpub}.
This suggested a rough dichotomy to
explain the borderline for decidability of (monadic)
boundedness problems for fragments of $\FO$.
On the positive side, our present results bring this approach
to fruition in the much wider and unifying setting of $\MSO$.
Part of this success draws on the above-mentioned
change of perspective, which allows us to re-chart the
relevant fragments with a decidable boundedness problem
into a taxonomy of relevant classes of structures
to which we can lift and extend our decidability results
from Part~I.

We link the more traditional approach to
the boundedness problem to this new perspective
in the following section\?: in particular, we discuss some of
the more prominent fragments that have featured in the quest
for decidability of boundedness so far, and review key results
from that tradition.

In Sections \ref{sect:transfer}~and~\ref{sect:interpretations} we discuss the natural model-theoretic
techniques that can be used to translate and extend our results\?:
transfer properties and reductions (Section~\ref{sect:transfer})
and interpretations (Sections~\ref{sect:interpretations}).
In view of the above discussion this yields results both
in terms of applicability of our key result to
wider classes of structures, and in terms of decidability results
for new fragments.

\medskip
\paragraph*{\itshape Proviso.}
In this part all vocabularies are (finite and) purely relational.

\section{Boundedness in the classical setting}
\label{sect:classical boundedness}
\label{sect:start II}

The key result concerning boundedness from classical
model theory is the following.

\begin{thm}[Barwise--Moschovakis \cite{BarwiseMoschovakis78}]
\label{thm:BMthm}
The following are equivalent for least fixed points
based on any $X$-positive $\varphi(X,\bar{x}) \in \FO$\?:
\begin{enumerate}
\item $\varphi$ is bounded.
\item $\varphi^\infty$ is uniformly $\FO$-definable.
\item $\varphi^\infty(\fA)$ is $\FO$-definable in each $\fA$.
\qed\end{enumerate}
\end{thm}

The classical proof is based on compactness arguments
and works with $\aleph_0$-saturated models for the crucial
implication from~(3) to~(1). It is immediate that this argument
relativises to natural fragments of $\FO$.
For formulae $\varphi$ from some such fragment of $\FO$ we
may replace $\FO$-definability by definability in the
fragment if that fragment has the natural closure properties
that render the finite stages definable\?; for truly natural
fragments like those to be considered below, however,
$\FO$-definability will imply definability within the
fragment by classical preservation theorems.

While these considerations offer some guidelines as to
what the right candidates $L \subseteq \FO$ for
decidable $\BDD(L)$ might be, our results from Part~I take
us beyond the limitations of $\FO$ and compactness --
which also means that boundedness becomes divorced from
definability of the fixed point.

\medskip
We start this section with a brief review of some logics and fragments
that feature prominently in connection with the boundedness problem --
be it in classical results or in new results flowing from our main theorem.
These may be grouped into three main categories\?:

\paragraph{\textit{Existential/universal fragments:}}
certain limited, purely existential/purely universal fragments
$\EFO \subseteq \FO$ and $\AFO\subseteq \FO$\?: these are the 
natural candidates
for a decidable monadic boundedness problem $\BDDm(L)$ in terms of quantifier
prefix classes $L \subseteq \FO$
(cf.\ the classical decision problem, \cite{BGG}).
For decidability of the boundedness problem
extra restrictions on the polarities of
the given relations,
which are statically used in the fixed-point recursion,
and on equality, are necessary.
See Section~\ref{subsec:existforall} below.

\paragraph*{\textit{Modal fragments:}}
the modal fragments of first-order and monadic second-order logic\?:
basic modal logic $\ML \subseteq \FO$ and its monadic fixed-point extension
$\Lmu \subseteq \MSO$, the bisimulation invariant
fragments of $\FO$ and $\MSO$, respectively.
See Section~\ref{subsec:modal} below.

\paragraph*{\textit{Guarded fragments:}} the corresponding but more general guarded fragments\?:
the basic guarded fragment $\GF \subseteq \FO$ and its fixed-point extension
$\muGF \subseteq \GSO$. These correspond to the fragments of
$\FO$ and guarded second-order logic $\GSO$, respectively,
that are invariant under guarded bisimulation.
With these logics we also
extend the scope of our discussion beyond monadic fixed points.
See Section~\ref{subsec:guarded} below.

\smallskip
In relation to $\BDD(L)$ or $\BDD(L,\cC)$ it is useful
to have in mind the following observation, which severely limits
the expectations regarding decidability but also points to
natural candidates.

\begin{observation}\label{obs:SAT vs BDD}
Assume that $\BDD(L)$ is non-trivial in the sense
that there are unbounded formulae $\varphi \in L$.
Then simple closure properties of~$L$ -- as for instance
closure under monadic relativisation and under conjunctions --
imply that the satisfiability problem $\SAT(L)$
reduces to the boundedness problem $\BDD(L)$.
An analogous reduction applies w.r.t.\ to restricted classes of models, i.e.,
for $\SAT(L,\cC)$ and $\BDD(L,\cC)$ provided
$\cC$ also satisfies some simple closure requirements --
as for instance closure under disjoint unions and trivial expansions by unary
predicates.
\end{observation}

We sketch one typical argument to this effect.
Fix some $\varphi(X,x) \in L$ that is unbounded.
Then a sentence
$\psi \in L$ is unsatisfiable if, and only if,
the formula $\varphi(X,x)^Q \wedge \psi^P$ is bounded\?;
here $\varphi(X,x)^Q$ and $\psi^P$ stand for the relativisations
to two distinct unary predicates $P$~and~$Q$,
which do not occur in either formula.
Clearly, unsatisfiability of~$\psi$ implies that
$\varphi(X,x)^Q \wedge \psi^P$ is unsatisfiable and hence
has closure ordinal $0$. Conversely, if $\psi$ is satisfiable, then
structures obtained as the disjoint union of a $P$-coloured model
of $\psi$ and a $Q$-coloured part show $\varphi(X,x)^Q \wedge \psi^P$
to be unbounded. The basic idea can be
modified to suit various other situations. For instance, for modal logic,
where disjoint unions are not the right choice,
one could look at boundedness for $\varphi^Q \wedge
\Diamond (P \wedge \psi^P)$ to decide satisfiability of $\psi$.

\medskip
We turn to the above-mentioned groups of logics.

\subsection{Purely existential and universal fragments}
\label{subsec:existforall}

$\EFO[\tau] \subseteq \FO[\tau]$ is the fragment of positive,
purely existential prenex first-order formulae (with equality), where
for $\BDD$ we also allow (positive occurrences of)
monadic second-order variables.
Dually, we let $\AFO[\tau] \subseteq \FO[\tau]$
be the fragment of prenex universal
first-order formulae that are negative in all
relation symbols from the underlying relational vocabulary
$\tau$ and equality, but of course we allow positive occurrences
of monadic second-order variables.

The first interest in boundedness as a decision problem
concerned the query language \textsc{Datalog}
corresponding to the evaluation of systems of least fixed
points of relational Horn clauses of the form
\begin{align*}
  \textstyle
  X\bar x \leftarrow \exists\bar y \bigwedge_i \alpha_i(\bar x,\bar y)
\end{align*}
with relational atomic formulae~$\alpha_i$.
This Horn clause translates into
\begin{align*}
  \textstyle
  \varphi(X,\bar x) = \exists\bar y \bigwedge_i \alpha_i(\bar x,\bar y) \in \EFO
\end{align*}
in our framework. In this connection
the first decidability results were obtained in
\cite{CosmadakisGaKaVa88}, and also the strict limitations for this
decidable case became apparent \cite{GaifmanMaSaVa93,HillebrandEtAl95}.

\begin{thm}\hfill
\begin{enumerate}[label=\({\alph*}]
\item
  The monadic boundedness problem $\BDDm(\EFO)$ is decidable \textup{\cite{CosmadakisGaKaVa88}.}
\item
  Boundedness for binary least fixed points in $\EFO$
  is undecidable\?; so is boundedness even for monadic least fixed points
  in the extension of\/ $\EFO$ that allows
  negated equalities (or negative and positive
  occurrences of some of the static relations)
  \textup{\cite{GaifmanMaSaVa93,HillebrandEtAl95}.}
\qed\end{enumerate}
\end{thm}

\noindent As for $\BDDm(\AFO)$, whose decidability was established in \cite{Otto06},
it should be noted that the fragment $\AFO$ is strictly dual to
$\EFO$\?; but as duality of fixed points links least to greatest
fixed points, trivial dualisation of the \textsc{Datalog} result would
just cover boundedness for greatest fixed points over $\AFO$.
Indeed, the techniques employed in \cite{Otto06} for decidability
of $\BDDm(\AFO)$ owe more to a reduction inspired by the guarded
fragment (see Section~\ref{subsec:guarded} below) and also do not seem to carry over directly to
$\BDDm(\EFO)$ or vice versa.

\begin{thm}[\cite{Otto06}]
$\BDDm(\AFO)$ is decidable, and both the restriction to monadic
least fixed points and the polarity restriction built into $\AFO$ are
necessary for decidability.
\qed\end{thm}

W.r.t.\ polarity restrictions on the static predicates in~$\tau$,
it should be noted that, as long as we consider the class of all
$\tau$-structures, it does not matter which polarity is prescribed,
since we can replace each predicate by its complement to switch
between polarities (this does not carry over from to $\BDD(L)$ to
$\BDD(L,\cC)$ unless $\cC$ is closed under predicate complementation).
What does matter, even over the class of all
$\tau$-structures, however, is whether we allow some predicates to
appear both positively \emph{and} negatively in $\varphi$.

\subsection{Logics of modal character}
\label{subsec:modal}

For a relational vocabulary $\tau$ consisting of only unary and
binary relation symbols, $\ML[\tau] \subseteq \FO[\tau]$ stands for the
\emph{modal fragment} of first-order logic. $\ML[\tau]$ is obtained
as the closure of monadic atomic formulae (where we also allow
monadic second-order variables besides unary relation symbols
in $\tau$) in a single free first-order variable under boolean
connectives and modal quantification of the form
\begin{align*}
  \psi(x) = \exists y ( Rxy \wedge \varphi(y))
  \quad\text{and, dually,}\quad
  \psi(x) = \forall y ( Rxy \to \varphi(y))
\end{align*}
for any $\varphi(y) \in \ML[\tau]$ and binary relation symbol $R \in \tau$.

The \emph{modal $\mu$-calculus} $\Lmu[\tau]$ is obtained as the natural
fixed-point extension of $\ML[\tau]$ through additional closure
under least fixed points\?: if $\varphi(X,x) \in \Lmu[\tau]$ is
positive in $X$, then $\psi(x) = \mu_X \varphi \in \Lmu[\tau]$
defines the least fixed-point $\varphi^\infty$.

\medskip
Our definition of $\ML$ is the usual embedding of basic
modal logic into $\FO$ by means of
the standard translation $\varphi \mapsto \varphi^*$,
which translates the modal
formula $\Box_R \varphi$ into
$(\Box\varphi)^* (x) = \forall y ( Rxy \to \varphi^*(y))$.
By van~Benthem's classical result in~\cite{Ben83},
$\ML[\tau]$ provides equivalent syntax for exactly those first-order
formulae in a single
free element variable whose semantics is preserved under bisimulation
equivalence. In this sense $\ML$
\emph{is} the bisimulation invariant (read\?: modal)
fragment of first-order logic. (For these and other basic facts
in the model theory of modal logic compare e.g.~\cite{GOHBML2007}).

We have similarly translated the $\mu$-calculus
in a manner that in particular turns it into a fragment of $\MSO$. In fact
$\Lmu$ \emph{is} the modal fragment of $\MSO$, in just the sense that
$\ML$ is the modal fragment of $\FO$, by an important result
of Janin and Walukiewicz \cite{JaninWalukiewicz}.

For us it will be important that $\ML \subseteq \Lmu \subseteq \MSO$ and that
$\ML$ and $\Lmu$ are preserved under bisimulation, which entails the
tree-model property.
Decidability of $\BDDm(\ML)$ was first shown in \cite{Otto99}\?;
note, however, that although that paper shows more generally that
it is decidable for an arbitrary formula of $\Lmu$ whether it is
equivalent to any formula in plain modal logic (of which
$\BDDm(\ML)$ is a special case, by the modal variant of the
Barwise--Moschovakis Theorem), it does \emph{not} deal with $\BDDm(\Lmu)$.

As will be reviewed in Section~\ref{sect:classical transfer} below,
decidability of
$\BDDm(\ML)$ and $\BDDm(\Lmu)$ can be essentially attributed to the tree-model
property stemming from bisimulation invariance.
Decidability of $\BDDm(\Lmu)$ is new here\?;
see Corollary~\ref{cor:decidability of BDD(GF), etc} below.
This result obviously implies the result of \cite{Otto99}
concerning decidability of $\BDDm(\ML)$ (but not as far as the problem
of equivalence of a given $\Lmu$-formula to some $\ML$-formula is concerned).

\begin{thm}
$\BDDm(\Lmu)$ and hence $\BDDm(\ML) \subseteq \BDDm(\Lmu)$ are decidable.
\end{thm}

\subsection{Guarded logics}
\label{subsec:guarded}

The \emph{guarded fragment} $\GF \subseteq \FO$ of
first-order logic extends the idea of
the local, relativised quantification of modal logic to the setting
of higher-arity relations. Since its inception in~\cite{ABN} the guarded
fragment and its extensions have been shown to mirror many of the nice
model-theoretic properties of modal logic in this more general setting.
Just like $\ML$ and its fixed-point extension $\Lmu$,
$\GF$ as well as its fixed-point extension $\muGF$
are decidable for satisfiability, cf.\ \cite{ABN,Graedel99,GW}.
Their roles as the guarded bisimulation invariant fragments of $\FO$
and a suitable guarded second-order logic are strictly analogous to those
of $\ML$ and $\Lmu$ as
bisimulation invariant fragments of $\FO$ and $\MSO$\?:
$\GF \subseteq \FO$ captures precisely those
$\FO$ definable properties that are preserved under guarded bisimulation
\cite{ABN}, and similarly for $\muGF \subseteq \GSO$ w.r.t.\
the natural guarded second-order logic $\GSO$, \cite{GHO}.
Like $\ML$, $\GF$ still has the finite model property, and
both $\GF$ and $\muGF$ have a generalised tree-model
property~\cite{Graedel99,GW}, which implies
in particular that every satisfiable formula of
$\muGF[\tau]$ is satisfiable in a model whose tree-width
is bounded by the width of $\tau$ (maximal arity of relations in
$\tau$). But note that $\muGF$ does not have the finite model
property, in fact this is already true of the extension of $\Lmu$
that admits modal operators along backward edges (inverse or past
modalities).
$\GF$ has long been considered a good candidate for decidability of
$\BDD(\GF)$.

\smallskip
Let us define these logics and the concept of guardedness in more detail.
A subset of a $\tau$-structure $\fA$ is called \emph{guarded}
if it is a singleton set or a set of the form
$\set{ a }{ a \in \bar{a} }$ for some $\bar{a} \in R^\fA$,
$R \in \tau$. Clearly the cardinality of guarded subsets in
$\tau$-structures is bounded by the width of $\tau$.
A tuple is guarded if the set of its components is contained
in some guarded subset.
A subset $W \subseteq A^r$ is called a \emph{guarded relation} over $\fA$
if all tuples $\bar{a} \in W$ are guarded in $\fA$.

Syntactically, a \emph{guard} for variables $\bar{x}$ is an atomic
formula $\alpha(\bar{x}) \in \FO[\tau]$ (relational atom or equality)
in which precisely the variables $x \in \bar{x}$ occur (as free variables).

\emph{Guarded quantification} is relativised first-order quantification
of the form
\begin{align*}
  \exists \bar y ( \alpha(\bar x) \land \varphi(\bar x) )
  \quad\text{and, dually,}\quad
  \forall \bar y ( \alpha(\bar x) \to \varphi(\bar x) )\,,
\end{align*}
where $\alpha$ is an atom (viz., a guard for $\bar{x}$),
$\mathrm{free}(\varphi) \subseteq \mathrm{free}(\alpha) = \set{x}{x \in \bar x}$
and $\bar y$ is any
tuple of (distinct) variables from $\mathrm{free}(\alpha)$.

\begin{defi}
\textup{(a)}
$\GF[\tau] \subseteq \FO[\tau]$, the \emph{guarded fragment}
of first-order logic, is obtained as the closure of
atomic $\FO[\tau]$-formulae under boolean connectives and
guarded quantification. We stress that, even if we admit a
second-order variable~$X$, $X$~may \emph{not} be used as
a guard for quantificational purposes.

\textup{(b)}
\emph{Guarded fixed-point logic} $\muGF$ is the natural extension
of $\GF$ that is additionally closed under the formation of
least fixed points over $X$-positive formulae.
Note again that second-order variables, which may
occur free or bound in formulae of $\muGF[\tau]$,
must not be used as guards.

\textup{(c)}
Also define \emph{strictly guarded} formulae of these logics to be those
formulae whose free first-order variables are explicitly guarded\?:
$\varphi(\bar x)$ is strictly guarded if it can only be satisfied by
guarded assignments to~$\bar x$ (a syntactic normal form can be
obtained with the help of the $\GF$-formula $\gdd(\bar x)$ below).
We denote these restrictions as $\GFs \subseteq \GF$ and $\muGFs \subseteq \muGF$.
\end{defi}

It is clear that $\ML \subseteq \GFs$ and $\Lmu \subseteq \muGFs$.
We also note in passing that there is, for every finite $\tau$ and arity $r$,
a $\GFs[\tau]$-formula
$\gdd(x_1,\dots, x_r)$ that uniformly defines the set of all
guarded $r$-tuples in $\tau$-structures $\fA$\?:
\begin{align*}
    \set{ \bar{a} \in A^r }{ (\fA,\bar{a}) \models \gdd(\bar{x}) }
  = \set{ \bar{a} \in A^r }{ \bar{a} \text{ guarded in } \fA }\,.
\end{align*}

Clearly these formulae can be used to restrict arbitrary
relations to their guarded parts.
For strictly guarded formulae we thus obtain a normal form of
\begin{align*}
  \gdd(\bar x) \wedge \varphi(\bar x)
\end{align*}
where $\bar x$ is the tuple
of all the free first-order variables of $\varphi$.

\smallskip
For \emph{guarded second-order logic} there are several formalisations,
which were shown to be equally expressive \emph{in the absence of free
second-order variables} in~\cite{GHO}.
As we shall see 
as a consequence of Theorems~\ref{thm:GFmuGF} and~\ref{thm:FOundec}
below, this equivalence breaks down
if free second-order variables (for the generation
of non-monadic least fixed points) are admitted.

Specifically, one can define $\GSO$ as the extension of either $\GF$
or $\FO$ by second-order quantifiers ranging over guarded relations.
This can be enforced syntactically by means
of the formulae $\gdd(\bar x)$ that uniformly
define the sets of all guarded $r$-tuples\?;
alternatively one can stick with ordinary
second-order syntax and modify the semantics
to admit just guarded relations as instantiations for
second-order variables (guarded semantics).
The equivalence between these two definitions according to
\cite{GHO} breaks down in the presence of free
second-order variables of arity greater than~$1$,
since such variables are not allowed to serve as guards.
Therefore, we introduce two variants of guarded second-order logic.
As we shall see below, the corresponding boundedness problems
are different\?: one is decidable for arbitrary fixed points,
while the other one is only decidable for monadic fixed points.

\begin{defi}
\emph{Guarded second-order logic} $\GSO[\tau]$ is the extension
of $\FO[\tau]$ by quantification over guarded relations.
We denote by $\GSOg[\tau]$ the fragment of $\GSO[\tau]$
where all first-order quantifications are guarded.

Again, we denote by $\GSOs$ and $\GSOgs$ the respective
fragments of \emph{strictly guarded} formulae,
in which the tuple of free first-order variables
is explicitly guarded.
\end{defi}

Clearly $\GF \subseteq \muGF \subseteq \GSOg \subseteq \GSO$.
Similar inclusions hold for the corresponding strict fragments.
Furthermore, $\MSO \subseteq \GSO$
since monadic relations are guarded (by the equality predicate).
We shall see that
the restriction to least fixed points that are guarded --
i.e., fixed points of formulae in the starred logics --
is the right counterpart, in the guarded world, for
monadic fixed points.
For the boundedness problem, moreover, we shall have
reductions from $\BDD(\GF)$ to $\BDD(\GFs)$ and
from $\BDD(\GSOg)$ to $\BDD(\GSOgs)$, see Section~\ref{GFredsec}.

\medskip
The guarded fragment $\GF$ as well as its fixed-point extension $\muGF$
are preserved under \emph{guarded bisimulation,} the infinitary game
equivalence associated to the restricted quantification pattern of
guarded quantification. Guarded bisimulation equivalence plays a role for
guarded logics that is analogous to the role of ordinary bisimulation
for modal logics. In fact, just as modal logic is
the bisimulation-invariant fragment of first-order logic \cite{Ben83},
so $\GF$ corresponds to the fragment of first-order logic
that is invariant under guarded bisimulation \cite{ABN}\?;
and just as $\Lmu$~is the bisimulation-invariant fragment
of monadic second-order logic \cite{JaninWalukiewicz},
so $\muGF$ corresponds to the fragment of $\GSO$ that is
invariant under guarded bisimulation \cite{GHO}.
Note that, despite its name, $\GSO$ is not invariant under guarded bisimulation.
The model theory and crucial algorithmic properties of $\GF$ and $\muGF$
are discussed in \cite{Graedel99} and \cite{GW}. For both logics,
much of their well-behavedness is due to invariance under
guarded bisimulation equivalence, and, consequently, the
`generalised tree-model property' \cite{Graedel99}\?:
by means of a natural process of guarded tree unfolding, any structure
can be transformed into a guarded bisimilar structure that admits a
tree-decomposition based on guarded subsets. Hence any satisfiable
formula of $\GF$ or $\muGF$ has a model which is guarded tree-decomposable
so that its tree-width is bounded by the width of the underlying vocabulary.

\medskip
Because of its vicinity to the modal fragment,
$\GF$ has been a promising candidate for
decidability of boundedness, even not just for
monadic least fixed points. Approaches to $\BDD(\GF)$
along those lines that worked for $\ML$ and even for $\AFO$
-- viz., the use of invariance under guarded bisimulation and
the guarded version of the Barwise--Moschovakis theorem --
have not been successful. Our present techniques do indeed
yield decidability of $\BDD(\GF)$,
see Corollary~\ref{cor:decidability of BDD(GF), etc},
and thus settle a major open problem in the classical
context. As we do not rely on either compactness or locality
criteria in our approach, we do get a much stronger result
in Theorem~\ref{thm: bdd for GSO over bounded twd},
concerning the decidability
of $\BDD(\GSOs,\cW_k)$, the boundedness problem for
least fixed points over $\GSOs$-formulae over the class
of all relational structures of tree-width up to~$k$.
This decidability is even uniform w.r.t.\ tree-width, so that 
both the $X$-positive $\GSOs$-formula and the tree-width 
parameter~$k$ may be regarded as input to a single algorithm.

\begin{thm}\label{thm:GFmuGF}
The following are decidable\?:
$\BDD(\GF)$, $\BDD(\muGF)$, $\BDD(\GSOg,\cW_k)$,
$\BDD(\GSOs,\cW_k)$, $\BDDm(\GSO,\cW_k)$.
\end{thm}

The transfer and reduction techniques to be discussed below
immediately show that decidability for $\BDD(\GFs)$ and
$\BDD(\muGFs)$ are an immediate
consequence of decidability for $\BDD(\GSOs,\cW_k)$.
These results essentially invoke the generalised tree-model
property of $\GF$. 

As far as undecidability results are concerned,
we have the following fundamental result,
which follows from the proof given in \cite{GaifmanMaSaVa93}.
\begin{thm}\label{thm:FOundec}
$\BDD(\FO,\cP)$ is undecidable,
where $\cP$~is the class of all finite paths.
\qed\end{thm}
\begin{cor}
$\BDD(\GSO,\cW_k)$ is undecidable.
\end{cor}
In the same way, we obtain undecidability of $\BDD(L,\cC)$
for every logic $L \supseteq \FO$ and class $\cC \supseteq \cP$
in which the class of all finite paths is $L$-definable.
Examples include boundedness of $\MSO$ over the class of all trees, over
the class of all finite trees, or over the class of all structures
of tree-width~$k$.

\medskip
The fragments discussed so far are closed under (at least)
positive boolean connectives and relativisation to unary
predicates. They are also closed under the substitution
operation used in defining the finite stages of fixed points.
So Observation~\ref{obs:SAT vs BDD}
applies to all of them and highlights the role of
$\EFO$, $\AFO$, $\ML$, $\Lmu$, $\GF$ and $\muGF$ as natural candidates
for decidability of $\BDD(L)$.
For $\FO$, $\MSO$ and $\GSO$ on the other hand,
not $\BDD(L)$ but at best
$\BDD(L,\cC)$ for suitably restricted classes $\cC$
can be decidable.

\section{Transfer properties for BDD}
\label{sect:transfer}

Model-theoretic transfer results involving special, restricted classes
of models are often useful. Key examples are provided by
the finite model property or the tree-model property, which, as
transfer results for satisfiability, can be useful
towards establishing decidability of $\SAT(L)$. The following
introduces a similar notion in connection with the boundedness
problem. The most far-reaching among these properties, which in the light
of our key result yields the strongest decidability consequences
for the boundedness problem, is the bounded-tree-width property.
We first define a general notion of transfer, then several
concrete specialisations that we need in the sequel.

\begin{defi}
A logic~$L$ allows \emph{$\cC$-to-$\cC'$ transfer for $\BDD$}
if, for all $\varphi \in L$,
$\varphi$~is bounded over~$\cC$ iff it is bounded over~$\cC'$\?:
$\BDD(L,\cC) = \BDD(L,\cC')$.

A logic~$L$ has the \emph{$\cC$-property for $\BDD$}
if it allows transfer from the class of all structures to $\cC$\?;
i.e., if $\BDD(L) = \BDD(L,\cC)$.
\end{defi}

Let $\cW_k$ stand for the class of all relational structures
of tree-width up to~$k$\?;
similarly $\cT_k$ stands for the class of tree models
of branching degree up to~$k$.

In accordance with the above, we say that
$L$~has the \emph{tree-width-$k$ property} for $\BDD$ for some concrete
bound~$k$ if $\BDD(L) = \BDD(L,\cW_k)$. In a similar spirit,
one could consider transfer properties from the class of all tree
models to the class of $k$-branching tree models, for concrete bounds $k$.
In both cases, however, our decidability arguments 
require just a computable dependence of the width parameter
on the input $\varphi \in L$, rather than a uniform constant bound.
This motivates the following.

\begin{defi}
We say that $L$ has the \emph{bounded-tree-width property for $\BDD$}
if, for some computable function~$f$,
$\varphi \in L$ is bounded iff $\varphi$~is bounded over $\cW_{f(\varphi)}$
(transfer to models of bounded tree-width).

Similarly, $L$ has the \emph{bounded-branching property for
$\BDD$ over trees} if, for some computable function~$f$,
$\varphi \in L$ is bounded over the class of all tree models
iff it is bounded over $\cT_{f(\varphi)}$
(transfer to tree models of bounded branching).
\end{defi}

In all natural cases a $\cC$-model property (transfer for $\SAT(L)$)
implies a $\cC$-property for $\BDD$. This is clearly the
case if $L$ is closed under the kind of substitution used to define
the finite stages and under boolean connectives. In that case,
the finite stages $\varphi^\alpha$ for $\alpha < \omega$ and the finite
stage increments $\varphi^{\alpha +1} \wedge \neg \varphi^\alpha$ are
definable by formulae in~$L$ and $\varphi$~is unbounded iff
all these formulae are satisfiable.

Concerning the finite model property for $\BDD$, note that
(even for fragments $L \subseteq \FO$) it does not
imply decidability of $\BDD(L)$\?: one still would need to check
satisfiability for each member of the infinite family
$\varphi^{\alpha+1} \wedge \neg \varphi^\alpha$ (albeit just in finite models).

\subsection{Transfer results for classical fragments}
\label{sect:classical transfer}

We collect some transfer results for the fragments
and logics discussed in the last section.

\begin{observation}\label{obs:classical transfer results}\hfill
\begin{enumerate}[label=\({\alph*}]
\item
  $\EFO$, $\AFO$, $\ML$, $\Lmu$ and\/ $\GF$
  have the \emph{finite model property} for $\BDD$ just as for $\SAT$.
\item[\normalfont(b)]
  $\ML$ and\/ $\Lmu$ have the \emph{tree-property} for
  $\SAT$ and\/ $\BDD$\?;
  $\ML$ even allows transfer to \emph{finite} tree-models
  of bounded branching.
\item
  $\EFO$, $\AFO$, $\ML$, $\Lmu$, $\GF$ and $\muGF$ all have
  the \emph{bounded-tree-width property} for $\SAT$ and\/ $\BDD$.
  Among these, the modal logics $\ML$, $\Lmu$ even allow transfer to
  tree models of bounded branching\?;
  $\EFO$, $\AFO$, $\GF$ and $\muGF$ allow transfer to models of
  bounded tree-width, in the case of $\EFO$, $\AFO$, $\GF$ even
  to \emph{finite} models of bounded tree-width.
\end{enumerate}
\end{observation}

\noindent More specifically, the necessary tree-width~$k$ in~(c) can be
bounded by the width of the underlying vocabulary~$\tau$
in the modal and guarded cases,
and (for a rough bound) by the size of the given formula~$\varphi$
in the case of $\EFO$, $\AFO$.

Most of these statements follow
from corresponding properties for $\SAT(L)$,
which are well known from the literature
(cf.~in particular Observation~\ref{obs:SAT vs BDD} above).
The bounded-tree-width property for $\BDD$ in the
case of $\GF$ and $\muGF$ is a direct consequence of
preservation of these logics under guarded bisimulation. Guarded
tree-unfoldings \cite{Graedel99,GHO} of arbitrary models yield
models possessing a tree decomposition whose bags are guarded subsets,
hence of width bounded by the width of~$\tau$.
For the assertions concerning the fragments $\EFO$ and $\AFO$, which
are not closed under negation, we prove the following lemma.

\begin{lem}
$\EFO[\tau]$ and $\AFO[\tau]$ allow transfer for $\BDDm$ to
finite models of bounded tree-width.\footnote{%
Here tree-width can be bounded by
the size of the given prenex formula $\varphi(X,x)$\?; a better bound
would be the tree-width of the quantifier-free kernel formula.}
\end{lem}
\begin{proof}
We explicitly treat the case of $\EFO$\?; the argument for $\AFO$ is
strictly analogous.

For $X$-positive $\varphi(X,x) \in \EFO[\tau]$ and finite $\alpha < \omega$,
the stage increment
$\varphi^{\alpha+1}(\fA) \smallsetminus \varphi^\alpha(\fA)$
is uniformly definable by a conjunction of a purely existential
formula
$\varphi^{\alpha+1}(x) \in \EFO[\tau]$
and a purely universal
formula in $\AFO[\tau]$ equivalent to the negation of $\varphi^\alpha(x)$.
Formulae of this kind are known to have the finite
model property\?:\footnote{They fall in particular within
the Bernays--Sch\"onfinkel class of prenex $\FO$-formulae
with quantifier prefix $\exists^\ast\forall^\ast$, cf.~\cite{BGG},
but a more direct argument suffices here.}
from an arbitrary model $(\fA,a)$ of some conjunction of a prenex
$\exists^\ast$-formula $\psi_1(x)$ and a prenex
$\forall^\ast$-formula $\psi_2(x)$, one obtains a finite model
by restricting $\fA$ to $a$ together with any chosen instantiation for the
existentially quantified variables in $\psi_1$\?; this restriction
still satisfies $\psi_1$, and as an induced substructure
of $(\fA,a) \models \psi_2$ it also still satisfies the universal
formula~$\psi_2$.

To obtain suitable (finite) models of bounded tree-width,
though, we need to consider the stronger preservation properties of
the formulae $\varphi^{\alpha}(x) \in \EFO[\tau]$,
and to some extent use the polarity constraints in $\EFO$ and $\AFO$.
The following argument also makes an interesting connection
with $\GF$.

Let w.l.o.g.\ $\varphi$ be of the form
\begin{align*}
  \varphi(X,x) =
  \exists \bar{y} \bigvee_i
  \bigl( \rho_i(x,\bar{y}) \wedge \bigwedge_{j \in s_i} X y_j \bigr)
\end{align*}
where $\bar{y} = (y_1,\dots, y_k)$,
the $\rho_i$ are conjunctions of relational $\tau$-atoms
(not involving $X$), and $s_i \subseteq \{ 1,\dots, k \}$.
For any $\tau$-structure $\fA$ let $\hat{\fA}$ be its expansion to
a $\hat{\tau}$-structure by new relations~$R_i$
of arity $k+1$, with $R_i$~defined by~$\rho_i$.
In $\hat{\fA}$, $\varphi$ is equivalent to the $\GFs$-formula
\begin{align*}
  \hat{\varphi}(X,x) =
  \exists \bar{y} \bigvee_i
  \bigl( R_i x\bar{y} \wedge \bigwedge_{j \in s_i} X y_j \bigr)\,.
\end{align*}

An analogous equivalence obtains for formulae $\varphi^\alpha(x) \in \GF[\tau]$
and $\hat{\varphi}^\alpha(x) \in \GF[\hat{\tau}]$
defining the finite stages w.r.t.\ $\varphi$ and $\hat{\varphi}$.

Obviously
\begin{align*}
  \bigwedge_i \forall x \forall \bar y \bigl( R_i x\bar y \to \rho_i(x,\bar y) \bigr)
  \models \forall x \bigl( \hat{\varphi}^\alpha(x) \to \varphi^\alpha(x) \bigr)\,,
\tag{$\ast$}
\end{align*}
where the
formula on the left-hand side is in $\GF[\hat{\tau}]$.
Note, however, that implications of the form
$\forall x \forall \bar{y} \bigl( \rho_i(x,\bar{y}) \to R_i x\bar{y} \bigr)$,
which would be needed
towards the equivalence between $\varphi^\alpha$ and $\hat{\varphi}^\alpha$
cannot in general be expressed in $\GF$.

Let $\hat{\fA}^\ast$ be a guarded bisimilar unfolding of~$\hat{\fA}$.
Its tree-width is bounded by the maximum of
the width of~$\tau$ and $k+1$.
We also write $\fA^\ast$ for the $\tau$-reduct of $\hat{\fA}^\ast$.
Let $\pi : \hat{\fA}^\ast \to \hat{\fA}$ be the
projection from the unfolding onto the base structure\?;
$\pi$~is a homomorphism inducing the natural guarded bisimulation
between $\hat{\fA}^\ast$ and $\hat{\fA}$.
Preservation of $\GF[\hat{\tau}]$ under guarded bisimulations implies
that, for all $\alpha < \omega$,
\begin{align*}
  \hat{\fA}^\ast, a \models \hat{\varphi}^\alpha
  \quad\text{iff}\quad
  \hat{\fA}, \pi(a) \models \hat{\varphi}^\alpha.
\end{align*}

Since $\fA,\pi(a) \models \varphi^\alpha$ implies $\hat\fA, \pi(a) \models \hat\varphi^\alpha$
and, therefore, also $\hat\fA^*,a \models \hat\varphi^\alpha$,
it follows with $(\ast)$ above that 
$\fA,\pi(a) \models \varphi^\alpha$ implies 
$\fA^*,a \models \varphi^\alpha$.

In the opposite direction,
since the~$\varphi^\alpha$, as existential positive
formulae, are preserved under homomorphisms,
the implication
$\fA^\ast, a \models \varphi^\alpha \Rightarrow
\fA, \pi(a) \models  \varphi^\alpha$ is straightforward.
Therefore, for all $a \in \fA^\ast$ and all $\alpha < \omega$,
\begin{align*}
  \fA^\ast, a \models \varphi^\alpha
  \quad\text{iff}\quad
  \fA, \pi(a) \models  \varphi^\alpha,
\end{align*}
whence $\cl{\varphi}{\fA} = \cl{\varphi}{\fA^\ast}$.
Hence $\varphi$ is bounded iff it is bounded over structures
whose tree-width is bounded by the
maximum of the width of~$\tau$ and $k+1$.
In order to get back to finite models
of bounded tree-width, we may apply the simple
argument from above to find a finite induced substructure within
some $(\fA^\ast, a)$ that still satisfies the corresponding
$\exists^\ast/\forall^\ast$-conjunction
$\varphi^{\alpha+1}(x) \wedge \neg \varphi^\alpha(x)$.
\end{proof}

\subsection{\boldmath Transfer for $\MSO$ over trees}

At the level of $\MSO$ we obtain a bounded-branching property
for $\BDDm$ over trees. The availability of transfer at least
down to countable branching is essential to make a connection
via interpretations with our core result that was formulated over
ternary trees.

\begin{prop}\label{prop:countable branching property for MSO}
$\MSO$ has a countable branching property for monadic $\BDDm$ over trees.
\end{prop}

This statement follows immediately from
Proposition~\ref{prop:Loewenheim-Skolem for trees} below,
whose proof relies on
the availability of tree automata for $\MSO$ and involves, as a key step,
a L\"owenheim--Skolem property for $\MSO$-theories of trees.
We employ a certain kind
of tree automata introduced by Walukiewicz~\cite{Walukiewicz02}.
\begin{defi}
An \emph{$\MSO$-automaton} is a tuple $\cA = \langle Q,\Sigma,q_0,\delta,\Omega\rangle$
with a finite set of states~$Q$, an input alphabet~$\Sigma$,
an initial state~$q_0$, a parity function $\Omega : Q \to \omega$,
and a transition function $\delta : Q \times \Sigma \to \MSO$
that, given a state~$q$ and a letter~$c$, returns an $\MSO$-formula
$\delta(q,c)$ over the signature $\set{ P_q }{ q \in Q }$.

Such an automaton takes a $\Sigma$-labelled directed tree
$t=(T,E,\lambda)$ as input.
A \emph{run} of~$\cA$ on~$t$ is a function $\varrho : T \to Q$
with $\varrho(r) = q_0$ for the root $r$ of $(T,E)$
such that
\begin{align*}
  \langle U_v,\bar P\rangle \models \delta(q,\lambda(v))\,,
  \quad\text{for all } v \in T\,,
\end{align*}
where the universe~$U_v$ of the structure is the set of all children of~$v$
and the unary predicates are $P_p := \varrho^{-1}(p) \cap U_v$.
The run~$\varrho$ is \emph{accepting}
if, and only if, for all infinite branches $v_0v_1\dots$ of~$(T,E)$
\begin{align*}
  \liminf_{n \to \infty} \Omega(\app{\varrho}{v_n}) \text{ is even.}
\end{align*}
The \emph{language recognised by~$\cA$} is the set $L(\cA)$
of all trees~$t$ such that there exists an accepting run
of~$\cA$ on~$t$.
\end{defi}

Over trees these automata have the same expressive power as monadic second-order logic.
\begin{thm}[Walukiewicz~\cite{Walukiewicz02}]
A class~$\cC$ of directed trees is definable
by an $\MSO$-sentence~$\varphi$ if, and only if,
it is recognised by some $\MSO$-automaton~$\cA$.
\qed\end{thm}

We use $\MSO$-automata to prove the following L\"owenheim-Skolem theorem.
\begin{prop}\label{prop:Loewenheim-Skolem for trees}
For every tree structure~$\fT$ there exists a countable tree structure $\fT_0 \subseteq \fT$
with the same $\MSO$-theory.
\end{prop}
\begin{proof}
We prove the proposition for directed trees.
Then the corresponding claim for undirected trees follows.
Suppose that $\fT$ is a directed tree with root~$r$.
Let us call a substructure $\fT_0 \subseteq \fT$ a
\emph{subtree} of~$\fT$ if $\fT_0$~is a tree and it contains the root~$r$.

To prove the claim, we construct a countable subtree $\fT_0 \subseteq \fT$
such that every $\MSO$-automaton accepting~$\fT$ also accepts~$\fT_0$.
Since every $\MSO$-formula is equivalent (on trees) to an $\MSO$-automaton
and since $\MSO$ is closed under complement, it follows that
$\fT$ and $\fT_0$ have the same $\MSO$-theory.

To construct~$\fT_0$ we proceed as follows.
For every $\MSO$-automaton~$\cA$ that accepts~$\fT$ and every vertex
$v \in T$, we fix a countable set $S_\cA(v) \subseteq T$ of children of~$v$
such that the following holds\?:
\begin{itemize}[label=$(*)$]
\item Every subtree $\fT_0 \subseteq \fT$ such that
  \begin{align*}
    v \in T_0 \quad\text{implies}\quad S_\cA(v) \subseteq T_0
  \end{align*}
  is accepted by~$\cA$.
\end{itemize}
Let us call a subtree~$\fT_0$ \emph{$S_\cA$-closed} if
$v \in T_0$ implies $S_\cA(v) \subseteq T_0$.
We take for~$T_0$ the minimal subset of~$T$ containing
the root~$r$ that is $S_\cA$-closed for every~$\cA$ accepting~$\fT$.
The subtree~$\fT_0$ induced by~$T_0$ is countable and has the desired property.

To define $S_\cA$ we fix an accepting run~$\varrho$ of~$\cA$ on~$\fT$.
Let $v \in T$ be a vertex with label~$c$ and let $U$~be the set of children
of~$v$ in~$T$.
For each state $q \in \varrho(v)$, $\varrho$~induces a structure $\langle U,\bar P\rangle$
satisfying the transition formula $\delta(q,c)$.
For every state $p \in Q$, we select a set $X^q_p \subseteq P_p = \varrho^{-1}(p) \cap U$ as follows.
If $P_p$~is countable, we set $X^q_p := P_p$.
Otherwise, we choose an arbitrary countably infinite subset $X_p \subseteq P_p$.
Then we set
\begin{align*}
  S_\cA(v) := \bigcup_{q \in Q} \bigcup_{p \in Q} X^q_p\,.
\end{align*}

We claim that, for every $S_\cA$-closed subtree $\fT_0 \subseteq \fT$,
the restriction of~$\varrho$ to~$T_0$ is an accepting run of~$\cA$ on~$\fT_0$.
Obviously, every infinite branch of~$\fT_0$ is an infinite branch of~$\fT$
and, hence, satisfies the parity condition. So we only need to check
that the transition formulae hold at each vertex.
Let $v \in T$ be a vertex with label~$c$ and with set of children~$U$, and
let $\langle U,\bar P\rangle$ be the structure induced by~$\varrho$.
Since $\varrho$~is a run, we have
\begin{align*}
  \langle U,\bar P\rangle \models \delta(\varrho(v),c)\,.
\end{align*}
Note that the structure $\langle U,\bar P\rangle$ has only unary relations.
There is a well-known Ehrenfeucht-Fra\"\i ss\'e argument showing that
an $\MSO$-sentence of quantifier rank~$m$ cannot distinguish
two such structures $\langle U,\bar P\rangle$ and $\langle U',\bar P'\rangle$,
provided that each quantifier-free $1$-type is realised the same
number of times in both structures, or it is realised at least $2^m$~times
in each structure.

By definition of~$S_\cA(v)$, it follows that,
for all subsets $U_0 \subseteq U$ containing $S_\cA(v)$,
the structures
\begin{align*}
  \langle U,\bar P\rangle
  \quad\text{and}\quad
  \langle U_0,\bar P|_{U_0}\rangle
\end{align*}
have the same $\MSO$-theory.
Consequently,
\begin{align*}
  \langle U_0,\bar P|_{U_0}\rangle \models \delta(\varrho(v),c)\,.
\end{align*}
In particular, this is the case for $U_0 := U \cap T_0$.
Therefore, $\varrho \restriction T_0$ is a run.
\end{proof}
\begin{rem}
If, instead of the full $\MSO$-theory, we are only interested in
the preservation of a single $\MSO$-sentence,
the construction of the theorem yields a tree that is finitely branching.
\end{rem}

\begin{proof}[Proof of Proposition~\ref{prop:countable branching property for MSO}]
Clearly, if an $\MSO$-formula $\varphi(X,x)$ is bounded over the class of all
trees, it is also bounded over the class of all countable trees.
Conversely, suppose that $\varphi(X,x)$ is unbounded over arbitrary
trees. Then we can find, for every $\alpha < \omega$, a tree~$\fT_\alpha$
satisfying the formula
$\psi_\alpha := \exists x[\varphi^{\alpha+1}(x) \land \neg\varphi^\alpha(x)]$.
By the above proposition, we can choose $\fT_\alpha$ to be countably branching.
Hence, $\varphi(X,x)$ is also unbounded over the class of all countably
branching trees.
\end{proof}

\section{Interpretations and reductions}
\label{sect:interpretations}

In the preceding section we have considered transfer of $\BDD(L,\cC)$
from one class~$\cC$ to a subclass $\cC_0 \subseteq \cC$.
In this section we will study more general reductions
of $\BDD(L,\cC)$ to $\BDD(L',\cC')$ where both the logic~$L$
and the class~$\cC$ may change.

\subsection{\boldmath A reduction for $\GF$}
\label{GFredsec}

We start by reducing $\BDD(\GSOg,\cC)$ to $\BDD(\GSOgs,\cC)$.
The following normal form for $\GSOg$-formulae is used in the proof
of the proposition below.
\begin{lem}\label{lem:normal form for GSOg}
Let $\varphi(\bar R,X,\bar x)$ be a $\GSOg$-formula with
free second-order variables $\bar R$,~$X$ and
free first-order variables~$\bar x$ that is positive in~$X$.
We can effectively construct 
$\GSOg$-formulae $\psi^0_i$, $\psi^1_i$, for $i < n$,
such that
\begin{align*}
  \varphi(\bar R,X,\bar x) \equiv \bigvee_{i < n}\bigl[\psi^0_i(X,\bar x) \land \psi^1_i(\bar R,X,\bar x)\bigr]\,,
\end{align*}
where
\begin{itemize}
\item the formulae~$\psi^0_i$ are quantifier-free and positive in~$X$,
\item the formulae~$\psi^1_i$ are positive in~$X$ and such that $X$~only appears in subformulae
  of the form $\forall\bar y\vartheta$ and $\exists\bar y\vartheta$.
\end{itemize}
Furthermore, if $\varphi$~is a $\GF$-formula,
then so are the formulae $\psi^0_i,\psi^1_i$, $i < n$.
\end{lem}
\begin{proof}
We may assume that $\varphi$~is in negation normal form.
The claim follows by induction on the structure of~$\varphi$.
All other cases being trivial, we present only the case of second-order quantifiers.

Hence, let us assume that $\varphi = \mathsf{Q}Z\vartheta$, for $\mathsf{Q} \in \{\forall,\exists\}$.
By inductive hypothesis, we may assume that
\begin{align*}
  \vartheta = \bigvee_{i < n}[\psi^0_i(X,\bar x) \land \psi^1_i(\bar R,Z,X,\bar x)]
\end{align*}
with $\psi^0_i$~and~$\psi^1_i$ as in the statement of the lemma.
In case of an existential quantifier, we are done since
\begin{align*}
  \exists Z\vartheta \equiv \bigvee_{i < n}[\psi^0_i(X,\bar x) \land \exists Z\psi^1_i(\bar R,Z,X,\bar x)]\,.
\end{align*}
For a universal quantifier, note that
\begin{align*}
  \forall Z\vartheta
  &= \forall Z\bigvee_{i < n}\bigl[\psi^0_i(X,\bar x) \land \psi^1_i(\bar R,Z,X,\bar x)\bigr] \\
  &\equiv \forall Z\bigwedge_{\sigma \in 2^n} \bigvee_{i < n} \psi^{\sigma(i)}_i \\
  &\equiv \bigwedge_{\sigma \in 2^n} \forall Z
          \Bigl[\bigvee_{i \in \sigma^{-1}(0)} \!\!\psi^0_i \ \lor\!\! \bigvee_{i \in \sigma^{-1}(1)} \!\!\psi^1_i\Bigr] \\
  &\equiv \bigwedge_{\sigma \in 2^n}
          \Bigl[\bigvee_{i \in \sigma^{-1}(0)} \!\!\psi^0_i \ \lor\ \forall Z\!\!\bigvee_{i \in \sigma^{-1}(1)} \!\!\psi^1_i\Bigr]\,.
\end{align*}
Hence, the claim follows by another application of the distributive law.
\end{proof}
\begin{prop}\label{prop: GSOg reduces to GSOgs}
For every formula $\varphi(X,\bar x) \in \GSOg[\tau]$,
we can effectively construct a formula $\varphi^g(X,\bar x) \in \GSOgs[\tau]$
such that $\varphi(X,\bar x)$ is bounded if, and only if, $\varphi^g(X,\bar x)$ is.

Furthermore, if $\varphi$~is a $\GF$-formula, then so is~$\varphi^g$.
\end{prop}
\begin{proof}
By the lemma we may assume that the formula $\varphi(X,\bar x)$ has the form
\begin{align*}
  \varphi(X,\bar x) = \bigvee_{i < n}\bigl[\chi_i(X,\bar x) \land \psi_i(X,\bar x)\bigr]\,,
\end{align*}
where the formulae~$\chi_i$ are quantifier-free
and in the formulae~$\psi_i$ every occurrence of~$X$ is in a subformula
of the form $\forall\bar y\vartheta$ and $\exists\bar y\vartheta$.

Note that any occurrence of an atom~$X\bar y$ that is in the
scope of some (guarded\?!) first-order quantification may be replaced by the formula
$X^g\bar y := X\bar y \land \gdd(\bar y)$ without changing its semantics.
Therefore,
\begin{align*}
  \varphi(X,\bar x) \equiv \bigvee_{i < n}\bigl[\chi_i(X,\bar x) \land \psi_i(X^g,\bar x)\bigr]\,,
\end{align*}
where $\psi_i(X^g,\bar x) := \psi_i(X^g/X,\bar x)$ is the formula obtained from~$\psi_i$
by replacing~$X$ by its guarded restriction~$X^g$ without affecting the semantics.
In the following a superscript $\rule{0pt}{1.5ex}^g$
is always used to indicate syntactic and/or semantic restriction to
the guarded part.

The fixed-point induction of $\varphi(X,\bar x)$
is closely related to the fixed-point induction of the strictly guarded formula
\begin{align*}
  \varphi^g(X,\bar x) :=
    \gdd(\bar x) \wedge \varphi(X,\bar x)
    \equiv \bigvee_{i<n} \bigl[\gdd(\bar x) \wedge \chi_i(X,\bar x)
                                            \wedge \psi_i(X^g,\bar x) \bigr]\,.
\end{align*}

Since $\varphi^g$~implies~$\varphi$ and both formulae are positive in~$X$,
it follows that the stages of~$\varphi^g$ are included in those for~$\varphi$.
In fact, it follows by a simple induction on~$\alpha$ that
$(\varphi^g)^\alpha(\fA) = (\varphi^\alpha(\fA))^g$.
Consequently, we have
\begin{align*}
  \cl{\varphi^g}{\fA} \leq \cl{\varphi}{\fA}
  \quad\text{and}\quad
  (\varphi^\infty(\fA))^g = (\varphi^g)^\infty(\fA)\,.
\end{align*}

If we can show that there exists a constant $n < \omega$, depending only on~$\varphi$,
such that, for all structures~$\fA$,
\begin{align*}
  \cl{\varphi}{\fA} \leq \cl{\varphi^g}{\fA} + n\,,
\end{align*}
then it follows that $\varphi$~is bounded if, and only if, $\varphi^g$~is bounded.

To find the constant~$n$, we consider the auxiliary formula
\begin{align*}
  \xi(Z,X,\bar x) :=
    \bigvee_{i<n} \bigl[\chi_i(X,\bar x) \wedge \psi_i(Z,\bar x)\bigr]
\end{align*}
in vocabulary $\tau \cup \{ Z \}$
(with a new second-order variable~$Z$ of the same arity~$r$ as~$X$,
which is regarded as a parameter) and
we consider its fixed-point induction in the expansion
$(\fA,P_0)$ of~$\fA$ where $Z$~is interpreted by the relation $P_0 := (\varphi^g)^\infty(\fA)$.
We claim that
\begin{align*}
  \cl{\varphi}{\fA} \leq \cl{\varphi^g}{\fA} + \cl{\xi}{(\fA,P_0)}\,.
\end{align*}

Let $\gamma := \cl{\varphi^g}{\fA}$. We have shown above that
$P_0 = (\varphi^g)^\gamma(\fA) \subseteq \varphi^\gamma(\fA)$.
Using monotonicity, it follows by a simple induction on~$\alpha$ that
\begin{align*}
  \varphi^\alpha(\fA) \subseteq \xi^\alpha(\fA,P_0) \subseteq \varphi^{\gamma+\alpha}(\fA)\,.
\end{align*}
The first inclusion implies that $\varphi^\infty(\fA) \subseteq \xi^\infty(\fA,P_0)$ while
the second inclusion implies that $\xi^\infty(\fA,P_0) \subseteq \varphi^\infty(\fA)$.
Setting $\beta := \cl{\xi}{(\fA,P_0)}$, it follows that
\begin{align*}
  \varphi^\infty(\fA) = \xi^\beta(\fA,P_0) \subseteq \varphi^{\gamma+\beta}(\fA) \subseteq \varphi^\infty(\fA)\,.
\end{align*}
Hence, $\cl{\varphi}{\fA} \leq \gamma + \beta$, as desired.

We have shown that, for every structure~$\fA$,
\begin{align*}
  \cl{\varphi^g}{\fA} \leq \cl{\varphi}{\fA} \leq \cl{\varphi^g}{\fA} + \cl{\xi}{(\fA,P_0)}\,.
\end{align*}
To conclude the proof it remains to prove that $\cl{\xi}{(\fA,P_0)}$ is uniformly bounded.
Note that $\xi$~treats~$Z$ as a static parameter, and only involves
its fixed-point variable~$X$ outside the scope of any quantifiers.
It follows that $\xi$~is trivially bounded (with a bound that is given by
the number of quantifier-free $r$-types in vocabulary $\tau \cup \{ Z \}$).
\end{proof}

Hence we may restrict attention
to fixed points over strictly guarded formulae.
This means that $\BDD(\GSOg)$ reduces to $\BDD(\GSOgs)$.
Let us remark that a corresponding result for $\GSO$ fails.

An argument analogous to the above
also applies to $\muGF$\?: according to \cite{GHO},
every $\muGF$-formula is equivalent to one where
every fixed-point operator is applied to a
strictly guarded formula.
For $\muGF$-formulae of this form,
a variant of Lemma~\ref{lem:normal form for GSOg} holds.
This is all we need for the proof of Proposition~\ref{prop: GSOg reduces to GSOgs}.
Consequently, $\BDD(\muGF)$ reduces to $\BDD(\muGFs)$.

\subsection{\boldmath $\MSO$-interpretations in trees}

In the first part we have obtained the decidability of $\BDDm(\MSO,\cT_3)$.
In this section, we use model-theoretic interpretations to
reduce the decidability of $\BDDm(\MSO,\cW_k)$ to this problem.

\begin{defi}
Let $\sigma$~and~$\tau$ be relational signatures.

\textup{(a)} A \emph{definition scheme} for an $\MSO$-interpretation
from~$\sigma$ to~$\tau$ is a list
\begin{align*}
  \cI = \bigl\langle\chi, \delta(x),\varepsilon(x,y),(\varphi_R(\bar x))_{R \in \sigma}\bigr\rangle
\end{align*}
of $\MSO[\tau]$-formulae where
$\chi$~is a sentence,
$\delta(x)$~has one free variable, $\varepsilon(x,y)$ has two,
and the number of free variables of $\varphi_R(\bar x)$ equals the arity of the
relation symbol~$R$.

\textup{(b)} The \emph{operation defined} by a definition scheme~$\cI$
maps $\tau$-structures~$\fA$ to $\sigma$-structures $\cI(\fA)$.
A $\tau$-structure~$\fA$ such that
$\fA \models \chi$, $\delta[\fA] \neq \emptyset$ 
and such that $\varepsilon$~defines an equivalence relation~$\sim$
on $\delta[\fA]$, 
is mapped to the $\sigma$-structure~$\fB$ with universe
\begin{align*}
  B := \set{ [a]_\sim \in A/{\sim} }{ \fA \models \delta(a) }
\end{align*}
and, for each $n$-ary relation $R \in \sigma$, the relation
\begin{align*}
  R^\fB := \set{ [\bar a]_\sim \in A^n/{\sim} }{ \fA \models \varphi_R(\bar a) }\,.
\end{align*}
For any other $\tau$-structure~$\fA$, we let $\cI(\fA)$ be undefined.

\textup{(c)} An \emph{$\MSO$-interpretation} is an operation defined
by a definition scheme~$\cI$.
If $\cC$~is a class of $\tau$-structures, we set
\begin{align*}
  \cI(\cC) := \set{ \cI(\fA) }{ \fA \in \cC \text{ such that } \cI(\fA) \text{ is defined} }\,.
\end{align*}
\end{defi}

For the proof of Proposition~\ref{prop: boundedness and interpretations} below,
let us recall the following well-known lemma.
We include a proof, so that we may refer to a
precise format of the formulae~$\psi^\cI$ later.
\begin{lem}[Interpretation Lemma]\label{lem:interpretations}
Let $\cI = \bigl\langle\chi, \delta(x),\varepsilon(x,y),(\varphi_R(\bar x))_{R \in \sigma}\bigr\rangle$
be an $\MSO$-interpretation.
For every $\MSO[\sigma]$-formula~$\psi$,
there exists an $\MSO[\tau]$-formula~$\psi^\cI$ such that,
for all $\tau$-structures~$\fA$ and every tuple $\bar a$ in~$A$, we have
\begin{align*}
  \fA \models \psi^\cI(\bar a)
  \quad\text{iff}\quad
  &\cI(\fA) \text{ is defined, }
    \fA \models \delta(a_i) \text{ for all } i, \text{ and}\\
  &\cI(\fA) \models \psi([\bar a]_\sim)\,. 
\end{align*}

If $\psi$~is positive in a predicate~$X$ and the
formula $\varphi_X(\bar x)$ from~$\cI$ is positive
in a predicate~$Y$, then $\psi^\cI$~is also positive in~$Y$.
\end{lem}
\begin{proof}
First, we define a formula~$\psi^*$ by induction on~$\psi$ as follows\?:
\begin{alignat*}{-1}
  (R\bar c)^* &:= \textstyle\exists\bar z\bigl[\bigwedge_i \varepsilon(z_i, c_i) \land \varphi_R(\bar z)\bigr]\,,
  &\qquad
  (\exists y\vartheta)^* &:= \exists y[\delta(y) \land \vartheta^*]\,, \\
  (c\seq d)^* &:= \varepsilon(c,d)\,,
  &\qquad
  (\forall y\vartheta)^* &:= \forall y[\delta(y) \to \vartheta^*]\,,
\end{alignat*}
and the translation ${}\cdot{}^*$ commutes with boolean operations and set quantifiers.

Then we can set
\begin{align*}
  \psi^\cI := \chi' \land \bigwedge_i \delta(x_i) \land \psi^*
\end{align*}
where
the conjunction is over all free variables of~$\psi$ and
$\chi' := \chi \land \eta$ is the conjunction of~$\chi$ with a formula~$\eta$ stating
that $\varepsilon$~defines an equivalence relation on~$\delta$.
\end{proof}

\begin{prop}\label{prop: boundedness and interpretations}
Let $\cI$~be an $\MSO$-interpretation and
$\cC$~a class of $\tau$-structures.
If\/ $\BDDm(\MSO,\cC)$ is decidable then so is $\BDDm(\MSO,\cI(\cC))$.
\end{prop}
\begin{proof}
We use the same notation as in the proof of Lemma~\ref{lem:interpretations}.
Suppose that $\cI = \langle\chi,\delta,\varepsilon,(\varphi_R)_R\rangle$
and let $\psi(x,X)$ be a formula over the signature $\sigma \cup \{x,X\}$.
We extend the notation~$\vartheta^*$ from above to formulae containing a free
set variable~$X$ by treating~$X$ as a relation defined by the formula~$Xx$, i.e.,
we set
\begin{align*}
  (Xc)^* := \exists z[\varepsilon(z,c) \land Xz]\,.
\end{align*}
Let $\fA \in \cC$ be a structure such that $\cI(\fA)$ is defined.
Note that Lemma~\ref{lem:interpretations} implies that
\begin{align*}
  \fA \models
    \forall x\forall y\bigl[\varepsilon(x,y) \to (\vartheta^*(x) \leftrightarrow \vartheta^*(y))\bigr]\,,
  \quad\text{for every formula } \vartheta(x)\,.
\end{align*}
For formulae $\vartheta(x)$~and~$\psi(X,x)$, it follows by induction on the structure of~$\psi$ that
\begin{align*}
  \fA \models
    \forall x\bigl[\delta(x) \to \bigl((\psi[\vartheta/X])^* \leftrightarrow \psi^*[\vartheta^*/X]\bigr)\bigr]\,.
\end{align*}
A simple induction on~$\alpha$ yields
\begin{align*}
  \fA \models \forall x\bigl[\delta(x) \to \bigl((\psi^\alpha)^* \leftrightarrow (\psi^*)^\alpha\bigr)\bigr]\,.
\end{align*}
Since $\fA \models \chi'$, it follows by the definition of the mapping $\vartheta \mapsto \vartheta^\cI$
that, for every $\alpha < \omega$, we have
\begin{align*}
  \cI(\fA) &\models \forall x(\psi^{\alpha+1} \leftrightarrow \psi^\alpha) \\
\text{iff}\quad
  \fA &\models \forall x(\psi^{\alpha+1} \leftrightarrow \psi^\alpha)^\cI \\
\text{iff}\quad
  \fA &\models \chi' \land \forall x\bigl[\delta(x) \to
               ((\psi^{\alpha+1})^* \leftrightarrow (\psi^\alpha)^*)\bigr] \\
\text{iff}\quad
  \fA &\models \forall x\bigl[(\chi' \land \delta(x) \land (\psi^{\alpha+1})^*) \leftrightarrow
                              (\chi' \land \delta(x) \land (\psi^\alpha)^*)\bigr] \\
\text{iff}\quad
  \fA &\models \forall x\bigl[(\chi' \land \delta(x) \land (\psi^*)^{\alpha+1}) \leftrightarrow
                              (\chi' \land \delta(x) \land (\psi^*)^\alpha)\bigr] \\
\text{iff}\quad
  \fA &\models \forall x\bigl[(\chi' \land \delta(x) \land \psi^*)^{\alpha+1} \leftrightarrow
                              (\chi' \land \delta(x) \land \psi^*)^\alpha\bigr] \\
\text{iff}\quad
  \fA &\models \forall x[(\psi^\cI)^{\alpha+1} \leftrightarrow (\psi^\cI)^\alpha\bigr]\,.
\end{align*}
Consequently, $\psi$~is bounded over $\cI(\cC)$
if and only if $\psi^\cI$ is bounded over~$\cC$.
\end{proof}

\begin{cor}\label{cor:decidability and subclasses}
Let $\cC$~be a class of $\tau$-structures and $\psi$~an $\MSO$-formula.
If\/ $\BDDm(\MSO,\cC)$ is decidable then so is
$\BDDm(\MSO,\cC_\psi)$ where
\begin{align*}
  \cC_\psi := \set{ \fA \in \cC }{ \fA \models \psi }\,.
\end{align*}
\end{cor}
\proof
We can use the interpretation $\cI = \langle\chi,\delta,\varepsilon,(\varphi_R)\rangle$
with
\begin{alignat*}{-1}
  \chi &:= \psi\,, &\qquad
  \varepsilon(x,y) &:= x\seq y\,, \\
  \delta(x) &:= x\seq x\,, &\qquad
  \varphi_R(\bar x) &:= R\bar x\,.\rlap{\hbox to 143 pt{\hfill\qEd}}
\end{alignat*}

For the application to boundedness below we will need the following
interpretation results. First, let us consider classes of trees.
The proof of the following lemma is straightforward.
For (a)~and~(b), we use the usual first-child/next-sibling encoding of a tree,
while for~(c) we use a marking of the root, which can be used to recover the
orientation of the edges since we can express reachability in $\MSO$.
\begin{lem}\label{lem: tree interpretations}
There exist $\MSO$-interpretations mapping
\begin{enumerate}
\item the class~$\cT_3$ of all ternary trees to the class~$\cT_{\aleph_0}$ of all countable trees\?;
\item the class of all finite ternary trees to the class of all finite trees\?;
\item the class of all undirected trees to the class of all directed trees.
\qed\end{enumerate}
\end{lem}

Next, we study structures of bounded tree-width.
\begin{defi}
Let $\tau$ be a relational vocabulary and $\fA$ a $\tau$-structure.
A \defn{tree-decomposition} of~$\fA$
is a $2^A$-labelled directed tree $D=\parlr{T,E,\lambda}$
satisfying the following conditions\?:
\begin{itemize}
\item $\bigcup_{t\in T} \lambda(t) = A$.
\item For all relation symbols $R\in\tau$ of arity~$n$ and all tuples
  $(a_1,\dots, a_n) \in R^\fA$,
  there is some $t\in T$ such that $a_1,\ldots,a_n \in \app{\lambda}t$.
\item For all $a\in A$,
  the set $\set{ t\in T }{ a\in\app{\lambda}t }$
  is connected in $\parlr{T,E}$.
\end{itemize}
The \defn{width} of~$D$ is $\max_{t\in T} {\card{\app{\lambda}t}} - 1$.
The \defn{tree-width} of~$\fA$ is the minimum width
of a tree decomposition of~$\fA$.
\end{defi}

\begin{lem}\label{lem: interpretation for bounded twd}
For every $k < \omega$ and all relational vocabularies~$\tau$,
there exists an $\MSO$-interpretation
mapping the class of all trees to the class $\cW_k[\tau]$
of all relational $\tau$-structures of tree-width at most~$k$,
and the class of all finite trees 
to the class of all finite structures from $\cW_k[\tau]$. 
\end{lem}
\proof
For finite trees, such an interpretation was first given
by Courcelle and Engelfriet~\cite{CourcelleEngelfriet95},
but using a slightly different notion of an interpretation.
We give a detailed proof, since the precise version needed here
does not appear in the literature.

We start by explaining how we can encode a structure $\fA \in \cW_k[\tau]$
into a tree.
Then we will construct an interpretation~$\cI$ performing the inverse translation.
Let $\fA$ be a $\tau$-structure of tree-width at most~$k$
and let $D=\parlr{T,E,\lambda}$ be a tree decomposition of~$\fA$ of width at most~$k$.
It is no restriction to assume that
no $\lambda(v)$ is empty
and that there is some injective function $\iota : A \to T$
with $a \in \lambda(\iota(a))$, for all~$a$.
For every $v\in T$, we
fix an enumeration $c^v_0,\dots,c^v_{\ell_v}$
of $\lambda(v)$ where $0 \leq \ell_v \leq k$
and $a = c^{\iota(a)}_0$ for all $a\in A$.
We obtain a tree structure
by turning $(T,E)$ into an undirected tree
expanded by the following monadic relations\?:
\begin{itemize}
\item a unary predicate~$R$ containing only the root\?;
\item unary relations~$E_{i,j}$, for $0\leq i,j\leq k$,
  containing those $v\in T$ such that $v$~has a parent~$u$ and $c^v_i = c^u_j$\?;
\item unary predicates $P_\fC$, for every $\tau$-structure~$\fC$ with universe
  $C = \{0,\ldots,\ell\}$, for some $0\leq\ell\leq k$,
  where $P_\fC$~contains those $v\in T$ such that $\fC$ is isomorphic
  to the substructure of $\fA$ induced by $\app{\lambda}v$
  via the isomorphism $i \mapsto c^v_i$.
\end{itemize}
The interpretation $\cI = \langle\chi, \delta(x),\varepsilon(x,y),(\varphi_R(\bar x))_{R \in \tau}\rangle$
that reverses this encoding is defined as follows.
The formula $\chi$~states that
\begin{itemize}
\item $R$ is a singleton, unless the universe is empty\?;
\item for all $0\leq i,i',j,j' \leq k$ with $i \neq i'$ and $j \neq j'$,
  $E_{i,j}$ is disjoint from $E_{i,j'}$ and from $E_{i',j}$\?;
\item the $P_{\fC}$ are disjoint\?;
\item if $v \in P_\fC$ and $0\leq i,j\leq k$ are such that $i\geq\card C$,
  then there is no $u \in T$ with $(v,u) \in E_{i,j}$ or $(u,v) \in E_{j,i}$.
\end{itemize}
For all $0\leq i,j\leq k$,
there is an $\MSO$-formula $\app{\varepsilon_{i,j}}{x,y}$
which defines the set of pairs $\parlr{v,u}$ such that $c^v_i = c^u_j$.
We set
\begin{align*}
  \delta(x) := x\seq x
  \quad\text{and}\quad
  \varepsilon := \varepsilon_{0,0}\,.
\end{align*}
Finally, for each relation symbol $R\in\tau$ of arity~$r$, we define
\[
  \app{\varphi_R}{x_1,\ldots,x_r} :=
    \bigvee_{0\leq i_1\leq k}\dots\bigvee_{0\leq i_r\leq k}
    \bigvee_{\substack{\fC \\ \parlr{i_1,\ldots,i_r}\in R^{\fC}}}
    \exists y
    \biggl[
      \bigwedge_{1\leq j\leq r}\app{\varepsilon_{0,i_j}}{x_j,y}
      \wedge P_{\fC}y
    \biggr]\,.\eqno{\qEd}
\]

\section{Decidability results for boundedness}
\label{sect:decidabilities}
\label{sect:end II}

Using the reduction techniques of the previous sections
we obtain a wealth of decidability results.
We start with $\BDDm(\MSO,\cT)$ and $\BDDm(\MSO,\cW_k)$.
\begin{prop}\label{prop:boundedness for all trees}
The monadic boundedness problem for $\MSO$ over the class of all trees is decidable.
\end{prop}
\begin{proof}
By Proposition~\ref{prop:Loewenheim-Skolem for trees}, an $\MSO$-formula
is bounded over the class of all trees if, and only if, it is bounded over
the class of all countable trees. Hence, it is sufficient to prove that
$\BDDm(\MSO,\cT_{\aleph_0})$ is decidable.
By Lemma~\ref{lem: tree interpretations},
there exists an $\MSO$-interpretation~$\cI$ mapping~$\cT_3$
to $\cT_{\aleph_0}$. Hence, the decidability of $\BDDm(\MSO,\cT_{\aleph_0})$
follows by Proposition~\ref{prop: boundedness and interpretations}
and Theorem~\ref{thm:boundedness for ternary trees}.
\end{proof}
\begin{thm}\label{theo:tree-width extension}
For every $k < \omega$ and all relational vocabularies~$\tau$,
$\BDDm(\MSO,\cW_k[\tau])$, the monadic boundedness problem for $\MSO$
over the class~$\cW_k[\tau]$ of all relational $\tau$-structures
of tree-width at most~$k$ is decidable.
\end{thm}
\begin{proof}
With the help of
Lemma~\ref{lem: interpretation for bounded twd}
and Proposition~\ref{prop: boundedness and interpretations}
we can reduce $\BDDm(\MSO,\cW_k[\tau])$ to $\BDDm(\MSO,\cT)$.
The latter is decidable by
Proposition~\ref{prop:boundedness for all trees}.
\end{proof}

\begin{cor}
$\BDDm(\EFO)$, $\BDDm(\AFO)$, and\/ $\BDDm(\ML)$ are decidable.
\end{cor}
\begin{proof}
For each of these logics,
Observation~\ref{obs:classical transfer results}
provides a transfer result
to (finite) structures of bounded tree-width.
\end{proof}

Using similar techniques as above, one can extend Theorem~\ref{theo:tree-width extension}
to the extension of $\MSO$ by counting quantifiers,
to guarded second-order logic $\GSO$, and to simultaneous fixed points.
Instead of replacing $\MSO$ by a stronger logic,
one can also replace tree-width by clique-width.

We only give a sketch of the proof.
Let us denote by $\MSO+\mathrm{C}$ and $\GSOs+\mathrm{C}$ the extension
of the respective logic by predicates of the form
$\lvert X\rvert < \aleph_0$ and $\lvert X\rvert \equiv k \pmod m$,
where $X$~is a second-order variable and $k,m < \omega$.
A \emph{simultaneous} fixed point is defined by a system of formulae
$\varphi_0(\bar X,\bar x_0),\dots,\varphi_{n-1}(\bar X,\bar x_{n-1})$
with first-order variables~$\bar x_i$ and $n$~second-order variables $X_0,\dots,X_{n-1}$.
\begin{thm}\label{thm: bdd for GSO over bounded twd}
For every $k < \omega$, the boundedness problem for simultaneous $(\GSOs + \mathrm{C})$-fixed points
over the class of all relational structures of tree-width at most~$k$
is decidable.
\end{thm}
\begin{proof}[Sketch]
Since structures of tree-width at most~$k$ are sparse,
we can find, for every $(\GSOs+\mathrm{C})$-formula, an equivalent $(\MSO+\mathrm{C})$-formula
(see \cite{Courcelle03,Blumensath10}).
Therefore, the boundedness problem reduces to the boundedness
of simultaneous $(\MSO+\mathrm{C})$-fixed points on that class.
Using the interpretation argument from above, we can reduce it further
to the boundedness for simultaneous $(\MSO+\mathrm{C})$-fixed points
on the class of all ternary trees. On ternary trees,
$\MSO+\mathrm{C}$ collapses to $\MSO$.
Therefore, we only need to decide boundedness for simultaneous
$\MSO$-fixed points.
Finally, using again an interpretation argument we can replace
a simultaneous fixed point by an ordinary one (by making several copies of each vertex
of the tree, one for each component of the simultaneous fixed point).
\end{proof}
\begin{cor}\label{cor:decidability of BDD(GF), etc}
The following problems are decidable\?:
$\BDD(\GF)$, $\BDD(\muGF)$, $\BDDm(\Lmu)$, $\BDD(\GSOg,\cW_k)$, and\/
$\BDDm(\GSO,\cW_k)$.
\end{cor}
\begin{proof}
By Observation~\ref{obs:classical transfer results},
$\GF$ and $\muGF$ have the bounded-tree-width property for $\BDD$.
Hence, $\BDD(\GF)$ and $\BDD(\muGF)$ reduce to $\BDD(\GF,\cW_k)$ and $\BDD(\muGF,\cW_k)$, respectively,
which in turn are subsumed by $\BDD(\GSOg,\cW_k)$.
According to Proposition~\ref{prop: GSOg reduces to GSOgs},
$\BDD(\GSOg,\cW_k)$ reduces to $\BDD(\GSOgs,\cW_k)$ which is decidable
by Theorem~\ref{thm: bdd for GSO over bounded twd}.

For $\BDDm(\GSO,\cW_k)$ note that, singletons being always guarded,
every $\GSO$-formula $\varphi(X,x)$
with a single free first-order variable~$x$ belongs
to $\GSOs$.
Hence, $\BDDm(\GSO,\cW_k)$ reduces to $\BDDm(\GSOs,\cW_k)$
and the claim follows again from Theorem~\ref{thm: bdd for GSO over bounded twd}.
\end{proof}

\bigskip
\section*{Part III. Complexity Results}

\medskip
\section{Complexity}
\label{sect:complexity}
\label{sect:start III}
\label{sect:end III}

In connection with our decision procedures we have not been
specific about the algorithmic complexities involved.
The fact that we have to deal with $X$-positive $n$-types
as basic data has a major impact on all upper bounds that
can be derived from our approach.
Space $\exp^n(\Theta(\len{\tau}))$ is necessary to even
store such a type ($\exp^n$ denotes
the $n$-fold application of the exponentiation operation,
that is, a tower of height~$n$).
Overall it is straightforward to check that, on input~$\varphi$,
our decision procedure runs
in time $\exp^{\app{\qr}{\varphi} + O(1)}(\len{\varphi})$.

We now provide a corresponding lower bound,
even for monadic boundedness for first-order logic over just finite trees.
Note that, for most natural fragments of $\MSO$,
one can obtain a lower bound from the complexity of the
satisfiability problem of the fragment.
For instance, $\BDDm(\ML)$ is \textsc{Pspace}-hard
since satisfiability for $\ML$ is \textsc{Pspace}-complete.
For first-order logic over finite words \emph{with order,}
as well as over finite trees \emph{without order,}
we can similarly derive lower bounds from corresponding 
bounds for the satisfiability problem.

\begin{thm}\label{theo:lower complexity bound}
\textup{(a)}
The boundedness problem $\BDDm(\FO,\cP)$ for 
first-order logic over the class of finite words with order, 
is complete for $\DSPACE(\exp^{\poly(n)}(1))$.

\textup{(b)}
The boundedness problem 
$\BDDm(\FO,\cT_\fin)$ for first-order logic over the class
of all finite trees is hard for $\DTIME(\exp^{\app{\poly}n}(1))$.
\qed\end{thm}

Part~(a) follows from the corresponding 
result for $\SAT(\FO,\cP)$\?;
see \cite{Reinhardt02} for a proof and exposition.
Part~(b) is a consequence of the following complexity bound for
$\SAT(\FO,\cT_\fin)$\?; although this is based on standard techniques,
we include a proof since this complexity bound does not seem to appear
in the literature.

\begin{prop}\label{Prop: SAT(FO,Tfin) hard}
$\SAT(\FO,\cT_\fin)$ is hard for $\DTIME(\exp^{\app{\poly}n}(1))$
under polynomial time reductions.
\end{prop}
\begin{proof}
We show that $\SAT(\FO,\cT_\fin)$ is hard for $\NTIME(\exp^{\app{\poly}n}(1))$,
which is the same as $\DTIME(\exp^{\app{\poly}n}(1))$.
We use the following tiling problem, which is complete for
$\NTIME(\exp^{\app{\poly}n}(1))$ (see~\cite{Harel85} for an overview)\?:
given a set~$D$ of tiles, two relations $H,V \subseteq D \times D$, and
a natural number~$n$ (in unary encoding), determine whether there exists a tiling
of the $(\exp^n(1)\times \exp^n(1))$-grid, i.e., a function
$\tau : \exp^n(1)\times \exp^n(1) \to D$ such that
\begin{align*}
  (\tau(x,y),\tau(x+1,y)) \in H
  \quad\text{and}\quad
  (\tau(x,y),\tau(x,y+1)) \in V\,,
  \quad\text{for all } x,y\,.
\end{align*}

For the reduction, we set $N := \exp^n(1)$.
One can show that the problem remains complete for
$\NTIME(\exp^{\app{\poly}n}(1))$
even if we require for convenience that there are at most~$N$ tiles.
Thus, we can represent tiles by numbers less than~$N$.
We construct a formula~$\psi$ that is satisfied by some finite tree
if, and only if, there exists a tiling of the $(N \times N)$-grid.

We use an encoding of numbers by directed trees
introduced in~\cite{FlumGrohe06} (see also~\cite{DawarGrKrSchw07})
where numbers from $\{0,\dots,N-1\}$ are encoded by trees of height at most~$n$.
The encoding is such that there are first-order formulae
$\varphi_N(x)$, $\varphi_{\min}(x)$, $\varphi_{\max}(x)$,
$\varphi_=(x,y)$ and $\varphi_{\mathrm{suc}}(x,y)$,
which can be constructed in time polynomial in~$n$,
with the property that
\begin{itemize}
\item a vertex~$v$ in a tree~$\fT$ satisfies $\varphi_N(x)$ if the
  subtree rooted at~$v$ encodes a number from $\{0,\dots,N-1\}$\?;
\item a vertex~$v$ satisfies $\varphi_{\min}(x)$ or $\varphi_{\max}(x)$
  if the subtree rooted at~$v$ encodes the number $0$~or $N-1$, respectively\?;
\item the formulae $\varphi_=(x,y)$ and $\varphi_{\mathrm{suc}}(x,y)$
  similarly define equality and the successor relation for numbers
  encoded in the subtrees rooted at~$x$ and~$y$.
\end{itemize}

\noindent We use this encoding to represent triples of numbers as follows.
The triple $(x,y,z)$ is encoded by a tree of the form
\begin{center}
\includegraphics{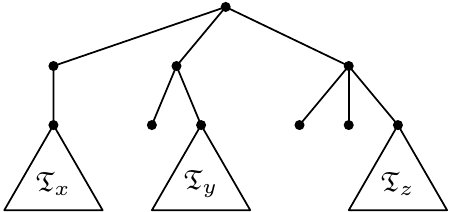}
%
%
%
%
%
%
%
\end{center}
where $\fT_x$, $\fT_y$, and $\fT_z$ are the trees representing
$x$, $y$,~and~$z$, respectively.
Then, a tiling can by represented by a set of triples $(x,y,z)$,
where $x$~and~$y$ are coordinates and $z$ is the tile at position $(x,y)$.
The respective set of trees is turned into a single tree
by making all these triples children of a new root.

To axiomatise the representation of a valid tiling,
we use a formula~$\psi$ based on the formulae
$\varphi_N$, $\varphi_{\min}$, $\varphi_{\max}$,
$\varphi_{\mathrm{suc}}$, and $\varphi_=$ from above.
The formula $\psi$ expresses the following\?:
\begin{itemize}
\item All children of the root encode triples of numbers.
\item There is some triple $(x,y,z)$ with $x = y = 0$.
\item Each triple has a neighbour to the right,
  unless the $x$-coordinate already is $N-1$.
  The tiles of the triple and its neighbour match.
\item Similarly, there is a neighbour above.
\item Each position occurs at most once.
\end{itemize}

Such a formula~$\psi$ is constructible
in time polynomial in $n$ and the size of the tile set.
Clearly, directed tree models of~$\psi$
correspond to valid tilings of an $(N\times N)$-grid.
Hence, $\psi$~is satisfiable by such a tree if, and only if,
such a tiling exists.

To work inside the class $\cT_\fin$,
we need to replace the directed trees by undirected ones.
Observe that every model of~$\psi$ is a tree of height at most $n+3$.
Hence, we can uniquely mark the root by attaching a path of length $n+4$ to it.
It is easy to modify~$\psi$ to work with such undirected trees instead.
\end{proof}

\begin{cor}
The following boundedness problems are $\app{\DTIME}{\exp^{\app{\poly}n}(1)}$-complete\?:
\begin{enumerate}
\item $\BDDm(\FO,\cT_{\mathrm{fin}})$ where $\cT_{\mathrm{fin}}$~is the class of all finite trees.
\item $\BDDm(\MSO,\cT)$ where $\cT$ is the class of all trees.
\item Boundedness for simultaneous $(\GSOs+\mathrm C)$-fixed points
  over the class of structures of bounded tree-width.
\end{enumerate}
\end{cor}
\begin{proof}
(1) follows from Proposition~\ref{Prop: SAT(FO,Tfin) hard}.
As (2)~reduces to~(3), for which we already have a trivial non-elementary upper
bound, it is sufficient to provide a lower bound for $\BDDm(\MSO,\cT)$.
As we have seen, $\BDDm(\FO,\cT_{\mathrm{fin}})$ reduces to
$\BDDm(\MSO,\cT_3)$, which in turn reduces to $\BDDm(\MSO,\cT)$.
Hence, the lower bound follows from~(1).
\end{proof}

This lower bound shows that, for many cases, our algorithm is best possible.
Of course, there are important fragments of $\MSO$ to which the lower bound is not applicable.
For instance, the following upper bounds are known from the literature\?:
\begin{thm}\leavevmode
\begin{enumerate}
\item $\BDDm(\ML)$ is in \textsc{Exptime} \textup{\cite{Otto99}.}
\item $\BDDm(\EFO)$ is in \textsc{2-Exptime} \textup{\cite{CosmadakisGaKaVa88}.}
\qed\end{enumerate}
\end{thm}
Since it is not the main concern of this article, we leave the exact
complexity of $\BDDm(\ML)$, $\BDDm(\Lmu)$, $\BDDm(\EFO)$, $\BDDm(\AFO)$, $\BDD(\GF)$, and $\BDD(\muGF)$ open.

\bibliography{bnd}
\bibliographystyle{plain}

\end{document}